\documentclass{llncs}
\usepackage[english]{babel}
\usepackage[utf8x]{inputenc}
\usepackage[T1]{fontenc}
\usepackage{fullpage}

\usepackage{enumitem}
\usepackage{caption}
\usepackage[skins]{tcolorbox}
\usepackage{amsmath}
\usepackage{mathtools}
\usepackage{graphicx}
\usepackage{amssymb}
\usepackage[colorinlistoftodos]{todonotes}
\usepackage[colorlinks=true, allcolors=blue,breaklinks=true]{hyperref}
\usepackage{breakcites}
\usepackage{upgreek}
\usepackage{bbold}
\usepackage{stmaryrd}
\usepackage{wasysym}
\usepackage{framed}
\usepackage{empheq}
\usepackage{enumitem}
\usepackage[boxed]{algorithm2e}
\usepackage{xcolor}
\usepackage{theorem}
\newcounter{counter}
\DeclareMathOperator*{\E}{\mathbb{E}}

\usepackage{float}
\usepackage[font=small,labelfont=bf]{caption}
\usepackage{xspace}

\newtheorem{defi}[counter]{Definition}
\newtheorem{thm}[counter]{Theorem}
\newtheorem{lem}[counter]{Lemma}
\newtheorem{ass}[counter]{Assumption}

\newtheorem{cor}[counter]{Corollary}

\theorembodyfont{\normalfont}
\newtheorem{rem}[counter]{Remark}

\newcommand{\bra}[1]{\langle #1|}
\newcommand{\ket}[1]{|#1\rangle}

\newcommand{\braket}[2]{\langle #1|#2\rangle}

\definecolor{dgreen}{rgb}{.1,.5,.1}

\newcommand{\proj}[1]{\ket{#1}\!\bra{#1}}

\usepackage{textgreek}
\newcommand{\sigp}{\textSigma-protocol\xspace}
\newcommand{\sigps}{\textSigma-protocols\xspace}

\def\regY{\ensuremath{\textit{\textsf{Y}}}}

\def\regZ{\ensuremath{\textit{\textsf{Z}}}}

\def\regE{\ensuremath{\textit{\textsf{E}}}}
\def\regX{\ensuremath{\textit{\textsf{X}}}}

\newcommand{\Gen}{\ensuremath{\mathsf{Gen}}\xspace}
\newcommand{\Sign}{\ensuremath{\mathsf{Sign}}\xspace}
\newcommand{\Ver}{\ensuremath{\mathsf{Verify}}\xspace}

\newcommand{\EUFNMA}{\mathsf{EUF\!-\!NMA}}
\newcommand{\sEUFCMA}{\mathsf{sEUF\!-\!CMA}}

\newcommand*\widefbox[1]{\fbox{\hspace{2em}#1\hspace{2em}}}

\DeclareMathSymbol{\shortminus}{\mathbin}{AMSa}{"39}

\newcommand{\switch}[2]{#1}  % Switching between LNCS style and narrow margins

\title{Security of the Fiat-Shamir Transformation\\in the Quantum Random-Oracle Model\thanks{\textcopyright IACR 2019. This is a reformatted and slightly modified version of the article submitted by the authors to the IACR and to
		Springer-Verlag in May 2019. The published version is available from the
		proceedings of  CRYPTO 2019.}}
\titlerunning{QROM security of the Fiat-Shamir Transformation}
\author{Jelle Don\inst{1,2} \and  Serge Fehr\inst{1,3,4} \and Christian Majenz\inst{2,4} \and Christian Schaffner\inst{2,4}}
\institute{
	Centrum Wiskunde \& Informatica (CWI), Amsterdam, Netherlands \and 
	Institute for Logic, Language and Computation, University of Amsterdam, Amsterdam, Netherlands \and 
	Mathematical Institute, Leiden University, Netherlands \and
	QuSoft, Amsterdam, Netherlands 
	\\ \email{jelle.don@cwi.nl}, \email{serge.fehr@cwi.nl}, \email{c.majenz@uva.nl}, \email{c.schaffner@uva.nl}}

\begin{document}
	\maketitle

	\begin{abstract}
		The famous Fiat-Shamir transformation turns any public-coin three-round interactive proof, i.e., any so-called \sigp, into a non-interactive proof in the random-oracle model. 
		We study this transformation in the setting of a {\em quantum adversary} that in particular may query the random oracle in quantum superposition. 
		
		Our main result is a generic reduction that transforms any quantum dishonest prover attacking the Fiat-Shamir transformation in the quantum random-oracle model into a similarly successful quantum dishonest prover attacking the underlying \sigp (in the standard model). Applied to the standard soundness and proof-of-knowledge definitions, our reduction implies that both these security properties, in both the computational and the statistical variant, are preserved under the Fiat-Shamir transformation even when allowing quantum attacks. 
		Our result improves and completes the partial results that have been known so far, but it also proves wrong certain claims made in the literature. 
		
		In the context of post-quantum secure signature schemes, our results imply that for any \sigp that is a proof-of-knowledge against quantum dishonest provers (and that satisfies some additional natural properties), the corresponding Fiat-Shamir signature scheme is secure in the quantum random-oracle model. 
		For example, we can conclude that the non-optimized version of \textsf{Fish}, which is the bare Fiat-Shamir variant of the NIST candidate \textsf{Picnic}, is secure in the quantum random-oracle model. 
	\end{abstract}

	\section{Introduction}
	
	\subsubsection{The (quantum) random-oracle model. }
	
	The {\em random-oracle model} (ROM) is a means to treat a cryptographic hash function $H$ as an ideal primitive. In the ROM, the only way to ``compute'' the hash $H(x)$ of any value $x$ is by making a {\em query} to an imaginary entity, the {\em random oracle} (RO), which has chosen $H$ uniformly at random from the set of {\em all} functions with the considered domain and range.
	
	The hope is that if a cryptographic scheme is secure in the ROM then it is also secure in the standard model, as long as $H$ is instantiated with a ``good enough'' cryptographic hash function. Even though in general we cannot hope to obtain provable security in the standard model in this way, (since there exist artificial counter examples~\cite{Canetti2004}), 
	this approach works extremely well in practice, leading to very efficient schemes that tend to resist all known attacks. 
	
	What makes the ROM particularly convenient is that in the security proof of a cryptographic scheme, we can control the RO. For instance, simply by recording the queries that the adversary makes to the RO, we know exactly which hash values he knows, and the hash value $H(x)$ is random to him for any $x$ that he has not queried. Furthermore, we can {\em reprogram} the RO, meaning that we can let $H(x)$ be some particular value $y$ for some specific $x$, as long as it is random from the adversary's perspective. 
	
	When considering a {\em quantum} adversary, the picture changes a bit. In order to model that such an adversary can evaluate any hash function {\em in superposition} on different inputs, we must allow such a quantum adversary in the ROM to make superposition queries to the RO: for any superposition $\sum_x \alpha_x \ket{x}$ it may learn $\sum_x \alpha_x \ket{x}\ket{H(x)}$ by making a single query to the RO.
	This is referred to as the {\em quantum random-oracle model} (QROM) \cite{Boneh2011}. 
	
	Unfortunately, these superposition queries obstruct the above mentioned advantages of the ROM. By basic properties of quantum mechanics one cannot observe or locally copy such superposition queries made by the adversary without disturbing them. Also, reprogramming is usually done for an $x$ that is queried by the adversary at a certain point, so also here we are stuck with the problem that we cannot look at the queries without disturbing them. 
	
	As a consequence, security proofs in the ROM almost always do not carry over to the QROM. This lack of proof does not mean that the schemes become insecure; on the contrary, unless there is some failure because of some other reason%
	\footnote{E.g., the underlying computational hardness assumption does not hold anymore in the context of a quantum adversary. },
	we actually expect typical schemes to remain secure. However, it is often not obvious how to find a security proof in the QROM. Some examples where security in the QROM has been established are \cite{Unruh2014,Zhandry2015,Eaton2015,Unruh2015,Kiltz2017,Alkim2017,Zhandry2018,Saito2018,Bos2018}.

	\subsubsection{Main technical result. }
	
	Our main technical result (Theorem~\ref{thmmain}) can be understood as a particular way to overcome\,---\,to some extent\,---\,the above described limitation in the QROM of not being able to ``read out'' any query to the RO and to then reprogram the corresponding hash value. Concretely, we achieve the following. 
	
	We consider an arbitrary quantum algorithm $\cal A$ that makes queries to the RO and in the end outputs a pair $(x,z)$, where $z$ is supposed to satisfy some relation with respect to $H(x)$, e.g., $z = H(x)$. We then show how to {\em extract} early on, by measuring one of the queries that $\cal A$ makes, the very $x$ that $\cal A$ will output, and to {\em reprogram} the RO at the point $x$ with a fresh random value $\Theta$, with the effect that the pair $(x,z)$ that $\cal A$ then outputs now satisfies the given relation with respect to~$\Theta$, with a not too large loss in probability. 
	
	The way this works is surprisingly simple. We choose the query that we measure uniformly at random among all the queries that $\cal A$ makes (also counting $\cal A$'s output), in order to (hopefully) obtain $x$. Subsequently we reprogram the RO, so as to answer $x$ with $\Theta$, {\em either} from this point on {\em or} from the following query on, where this binary choice is made at random. This last random decision seems counter-intuitive, but it makes our proof work. Indeed, we prove that the probability that $(x,z)$ satisfies the required relation drops by no more than a factor $O(q^2)$, where $q$ is the number of oracle queries $\cal A$ makes.

	\subsubsection{Application to the Fiat-Shamir transformation. }
	
	The Fiat-Shamir transformation~\cite{Fiat1987} turns any public-coin three-round interactive proof, i.e., any so-called \sigp, into a non-interactive proof in the (Q)ROM. In the classical case it is well known that the security properties of the \sigp are inherited by the Fiat-Shamir transformation~\cite{Bellare1993a,Faust2012}. In the quantum setting, when considering the security of the Fiat-Shamir transformation against quantum dishonest provers in the QROM, mainly negative results are known\,---\,see below for a more detailed exposition of previous results and how they compare to ours. 
	
	It is quite easy to see that the above result on the reprogrammability of the RO is exactly what is needed to turn a quantum prover that attacks the Fiat-Shamir transformation into a quantum prover that attacks the underlying $\Sigma$ protocol. Indeed, from any Fiat-Shamir dishonest prover $\cal A$ that tries to produce a proof $\pi = (a,z)$ for a statement $x$, we obtain an interactive dishonest prover for the $\Sigma$ protocol that extracts $a$ from $\cal A$ and sends it to the verifier, and then uses the received challenge $c$ to reprogram the RO, so that the $z$ output by $\cal A$ will be a correct reply with respect to $c$ with a probability not much smaller than the probability that $\cal A$ succeeds in forging $\pi$ in the QROM. 
	
	This gives us a very generic transformation (stated in Theorem~\ref{thmFS} below) from a Fiat-Shamir dishonest prover to a \sigp dishonest prover that is similarly successful, up to a loss in probability of order $O(q^2)$. Applied to the standard notions of soundness and proof-of-knowledge, we prove that both these security properties, in both the computational and the statistical variant, are preserved under the Fiat-Shamir transformation in the QROM (Corollaries~\ref{corProof} and~\ref{corPoK}).

	\subsubsection{Comparison with prior results. }
	
	Mainly negative results are known about the security of the Fiat-Shamir transformation against quantum attacks. ~\cite{Ambainis2014} presented attacks (relative to some oracle) against \sigps that satisfy only {\em computational} unique responses (see Section \ref{seccollapsing} for an informal definition) as opposed to {\em perfect} unique responses, and showed that these attacks carry over to a Fiat-Shamir transformed protocol.
	
	Currently, the only known positive result on the security of the Fiat-Shamir transformation against quantum attacks is the result by Unruh~\cite{Unruh2017}, which shows that statistical soundness of the \sigp implies statistical soundness of the Fiat-Shamir transformation.\footnote{In the (quantum) random-oracle model, {\em statistical} security considers a computationally unbounded attacker with a polynomially bounded number of oracle queries. \cite{Unruh2017} in contrast refers to this setting as `computational'.} One of the implications of our result is that this carries over to {\em computational} soundness: if the \sigp is {\em computationally sound} (as a `proof'), 
	% (but not necessarily computationally {\em special} sound) 
	then its Fiat-Shamir transformation is computationally sound as well. 
	Interestingly, Unruh seems to suggest in \cite{Unruh2017} (right after Theorem 21)  that this is not true in general, due to a counterexample from \cite{Ambainis2014}. The counter example is, however, a \sigp that is computationally {\em special} sound but not computationally sound (the issue being that in the quantum setting, special soundness does not imply ordinary soundness). 
	
	We point out that \cite{Dagdelen} claims an impossibility result about the soundness of the Fiat-Shamir transformation as a quantum proof of knowledge, which contradicts the implications of our result to proofs of knowledge. However, their result only applies to a restricted notion of proof of knowledge where the extractor is not allowed to measure any of the adversary's queries to the random oracle. The rational for this restriction was that such a measurement would disturb the adversary's quantum state beyond control; however, our technical result shows that it actually is possible to measure one of the adversary's queries and still have sufficient control over the adversary's behavior. 
	
	Indeed, our generic transformation from a Fiat-Shamir dishonest prover to a \sigp dishonest prover implies that {\em any} security property (against dishonest provers) of the \sigp carries over  unchanged to the Fiat-Shamir transformation, be it computational or statistical, plain soundness or the stronger proof-of-knowledge property.

	\subsubsection{Circumventing prior negative results. }
	
	At first glance, the negative results from \cite{Ambainis2014} against \sigps together with our new positive results seem to give a complete answer to the question of the security of the Fiat-Shamir transformation against quantum attacks. However, there is actually more to it. 
	
	We consider a {\em stronger} but still meaningful notion of {\em computationally unique responses}, which is in the spirit of the {\em collapsing property} as introduced by Unruh \cite{Unruh2016}. We call the new notion \emph{quantum computationally unique responses} and define it in Definition \ref{def:qcur}. Adapting a proof from \cite{Unruh2012}, it is not hard to see that a \sigp with (perfect or computational) special soundness and quantum computationally unique responses is a computational proof of knowledge. Therefore, our main result then implies that its Fiat-Shamir transformation is a computational proof of knowledge as well.
	
	Thus, with the right adjustments of the considered {\em computational} soundness properties, the negative results from~\cite{Ambainis2014} may actually be turned into positive answers. One caveat here is that we expect proving quantum computationally unique responses to be much harder than computationally unique responses.

	\subsubsection{Application to signatures. }
	
	Our positive results on the Fiat-Shamir transformation have direct applications to the security of Fiat-Shamir signatures. From the proof-of-knowledge property of the Fiat-Shamir transformation we immediately obtain the security of the Fiat-Shamir signature scheme under a {\em no-message attack}, assuming that the public key is a hard instance (Theorem~\ref{thm:NMA}). Furthermore, \cite{Unruh2017} and \cite{Kiltz2017} have shown that for Fiat-Shamir signatures, up to some loss in the security parameter and under some additional mild assumptions on the underlying \sigp, one can also derive security under {\em chosen-message attack}. 
	
	In conclusion, Fiat-Shamir signatures offer security against quantum attacks (in the QROM) if the underlying \sigp is a proof of knowledge against quantum attacks and satisfies a few additional natural assumptions (Theorem~\ref{thm:CMA}). 
	
	As a concrete application, assuming the hash function used in creating the commitments is collapsing, we can conclude that the non-optimized version of \textsf{Fish}, which is the Fiat-Shamir variant of \textsf{Picnic}, is secure in the QROM.
	
	\subsubsection{Comparison with concurrent results. }
	
	In concurrent and independent work\switch{~}{\ }\cite{LZ19}\footnote{The paper~\cite{LZ19} was put on eprint  ({\tt ia.cr/2019/262}) a few days after our eprint version ({\tt ia.cr/2019/190}). }, Liu and Zhandry show results that are very similar to ours: they also show the security of the Fiat-Shamir transformation in the QROM, and they introduce a similar stronger version of the computational unique responses property in order to argue that a \sigp is a (computational) proof of knowledge against a quantum adversary. 
	In short, \cite{LZ19} differs from the work here in the following aspects. In~\cite{LZ19}, the result on the Fiat-Shamir transformation is obtained using a very different approach, resulting in a greater loss in the reduction: $O(q^9)$ compared to the $O(q^2)$ loss that we obtain. On the other hand, on the quantum proof of knowledge front, Liu and Zhandry introduce some additional techniques that, for instance, allow them to prove that the \sigp underlying \textsf{Dilithium} satisfies (their variant) of the newly introduced strong version of the computational unique responses property, while we phrase this as a conjecture in order to conclude the security of (some variant of) the \textsf{Dilithium} signature scheme.

	\section{Reprogramming the Quantum Random Oracle}\label{secRepr}

	We show and analyze a particular way to reprogram a random oracle in the quantum setting, where the oracle can be queried in superposition. 
	% Concretely, we consider an arbitrary quantum algorithm $\cal A$ that makes black-box queries to a random oracle $H$ and that outputs a (possibly randomized) string $x$ together with a (possibly quantum) output $z$ that satisfies some property with respect to $H(x)$ with a certain probability. We show that it is possible to {\em extract} the value of $x$ that $\cal A$ will output by measuring one of the queries of $\cal A$, and to {\em reprogram} $H(x)$ to a fresh random value $\Theta$ so that with a related probability, $\cal A$ will now output $x$ together with a $z$ that satisfies the required property with respect to~$\Theta$. We show that our extract-and-reprogram procedure decreases the success probability of $\cal A$ by at most $O(q^2)$, where $q$ is the number of queries. 

	% We prove our main result: Suppose that a reduction algorithm measures a randomly chosen QROM-query, and subsequently reprograms the oracle on the measurement outcome. The simulation of the adversary may then continue as normal, except that the adversary is now forced to use the measurement outcome and the reprogrammed value in its forgery, which - up to polynomial factors - has the same success probability as an undisturbed run.
	\subsection{Notation}
	% We define some formal notation used throughout the proof of our main theorem.
	
	We consider a quantum oracle algorithm $\cal A$ that makes $q$ queries to an {\em oracle}, i.e., an unspecified function $H: {\cal X} \to {\cal Y}$ with finite non-empty sets ${\cal X},{\cal Y}$.   
	We may assume without loss of generality that $\cal A$ makes no intermediary measurements. %Any final measurements can be assumed to be part of the success criterion. 
	Formally, $\cal A$ is then described by a sequence of unitaries $A_1,\ldots,A_q$ and an initial state $\ket{\phi_0}$.%
	\footnote{Alternatively, we may understand $\ket{\phi_0}$ as an auxiliary input given to $\cal A$. }
	The unitaries $A_i$ act on registers $\regX,\regY,\regZ,\regE$, where $\regX$ and $\regY$ have respective $|{\cal X}|$- and $|{\cal Y}|$-dimensional state spaces, while $\regZ$ and $\regE$ is arbitrary. As will become clear, $\regX$ and $\regY$ are the quantum registers for the queries to $H$ as well as for the final output~$x$, $\regZ$ is for the output $z$, and $\regE$ is internal memory.
	For any concrete choice of $H: {\cal X} \to {\cal Y}$, we can write 
	$$
	{\cal A}^H \ket{\phi_0} := A_q\mathcal{O}^H \cdots A_1\mathcal{O}^H \ket{\phi_0} \, ,
	$$ 
	for the execution of $\cal A$ with the oracle instantiated by $H$, where $\mathcal{O}^H$ is the unitary $\mathcal{O}^H : \ket{x}\ket{y} \mapsto \ket{x}\ket{y \oplus H(x)}$ that acts on registers $\regX$ and $\regY$. 
	
	It will be convenient to introduce the following notation. For $0 \leq i,j \leq q$ we set 
	$$
	\mathcal{A}_{i\rightarrow j}^H := A_{j}\mathcal{O}^H \cdots A_{i+1}\mathcal{O}^H
	$$
	with the convention that $\mathcal{A}_{i\rightarrow j}^H := \mathbb{1}$ for $j \leq i$. Furthermore, we set
	$$
	\ket{\phi_i^H} := \big(\mathcal{A}_{0\rightarrow i}^H\big)\ket{\phi_0} 
	$$ 
	to be the state of $\cal A$ after the $i$-th step but right before the $(i+1)$-st query, and so that $\ket{\phi_q^H}$ equals $\big({\cal A}_{0\rightarrow q}^H\big)\ket{\phi_0}$ = ${\cal A}^H \ket{\phi_0}$, the output state produced by $\cal A$. 
	
	Finally, for a given function $H: {\cal X} \to {\cal Y}$ and for fixed $x \in {\cal X}$ and $\Theta \in {\cal Y}$, we define the {\em reprogrammed} function $H\!*\!\Theta x: {\cal X} \to {\cal Y}$ that coincides with $H$ on ${\cal X} \setminus \{x\}$ but maps $x$ to $\Theta$. With this notation at hand, we can then write
	$$
	\big(\mathcal{A}_{i\rightarrow q}^{H*\Theta x}\big) \, \big(\mathcal{A}_{0\rightarrow i}^{H}\big) \, \ket{\phi_0} = \big(\mathcal{A}_{i\rightarrow q}^{H*\Theta x}\big)\ket{\phi_i^H}
	$$
	for an execution of $\cal A$ where the oracle is reprogrammed at a given point $x$ after the $i$-th query.

	% Since we are merely interested in the $\regY \regZ$-component of the output state $\ket{\phi_q^H}$, as well as for some technical reasons, we will require $\ket{\phi_q^H}$ to be $\ket{0}$ within register $\regW$, no matter how $H$ is chosen. This can always be achieved by an insignificant modification to $\cal A$, i.e., by swapping $\regW$ with a default register within $\regE$. 
	
	We are interested in the probability that after the execution of ${\cal A}^H$ and upon measuring register $\regX$ in the computational basis to obtain $x \in {\cal X}$, the state of register $\regZ$ is of a certain form dependent on $x$ and $H(x)$.  
	This relation is captured by a projection $G_x^H$, where, more generally, for $x,x' \in \cal X$ and $\Theta \in \cal Y$ we set
	$$
	G_{x,x'}^{\Theta} = \proj{x'} \otimes \mathbb{1} \otimes \Pi_{x,\Theta} \otimes \mathbb{1} \, ,
	$$
	where $\{\Pi_{x,\Theta}\}_{x \in {\cal X},\Theta \in {\cal Y}}$ is a family of projections acting on $\regZ$, which we refer to as a \emph{quantum predicate}. We use the short hands $G_x^{\Theta}$ for $G_{x,x}^{\Theta}$ and $G_x^H$ for \smash{$G_x^{H(x)}$}, i.e., 
	$$
	G_x^{H} = \proj{x} \otimes \mathbb{1} \otimes \Pi_{x,H(x)} \otimes \mathbb{1} \, .
	$$
	% Above, $\Pi_{x,\Theta}$ is an arbitrary projection depending on $x$ and $\Theta$, acting on $\regZ$. 
	For an arbitrary but fixed $x_\circ \in {\cal X}$, we then consider the probability
	$$
	\|G_{x_\circ}^H \ket{\phi_q^H}\|_2^2 \, .
	$$
	Understanding ${\cal A}^H$ as an algorithm that outputs the measured $x$ together with the state $z$ in register $\regZ$, we will denote this probability also by
	$$
	\Pr\bigr[x\!=\! x_\circ \wedge V(x,H(x),z) : (x,z) \leftarrow {\cal A}^H \bigl] \, ,
	$$
	understanding $V$ to be a quantum predicate specified by the projections $\Pi_{x,H(x)}$. 

	\subsection{Main technical result}\label{secmainresult}
	
	We consider a quantum oracle algorithm $\mathcal{A}$ as formalized above, and we define a two-stage algorithm $\mathcal{S}$ with black-box access to $\mathcal{A}$ as follows. In the first stage, $\mathcal{S}$ tries to predict $\mathcal{A}$'s future output $x$, and then, upon input a (random) $\Theta$, in the second stage tries to output what $\mathcal{A}$ is supposed to output, but now with respect to $\Theta$ instead of $H(x)$. 
	
	$\mathcal{S}$ works by running $\mathcal{A}$, but with the following modifications. First, one of the $q+1$ queries of $\mathcal{A}$ (also counting the final output in register $\regX$) is selected uniformly at random and this query is measured, and the measurement outcome $x$ is output by (the first stage of) $\mathcal{S}$. Then, this very query of $\cal A$ is answered either using the original $H$ {\em or} using the reprogrammed oracle $H*\Theta x$, with the choice being made at random, while all the remaining queries of $\cal A$ are answered using oracle $H\!*\!\Theta x$.%
	\footnote{If it is the final output that is measured then there is nothing left to reprogram. } 
	Finally, (the second stage of) $\cal S$ outputs whatever $\mathcal{A}$ outputs.
	
	Here, the figure of merit is the probability that
	%old formulation: 
	% $x$ is extracted {\em and} the final output consists of the same $x$ together with a quantum state 
	%new formulation
	for a fixed $x$,  both the intermediate measurement and a measurement of the register $\regX$ return  $x$ \emph{and} that the register $\regZ$ contains a state
	%end new formulation 
	that satisfies the considered quantum predicate with respect to $x$ and its (now reprogrammed) hash value $\Theta$. Formally, this probability is captured by
	
	$$
	\E_{\Theta,i,b} \Bigl[\big\|G_x^\Theta \, \big(\mathcal{A}_{i+b\rightarrow q}^{H*\Theta x}\big) \, \big(\mathcal{A}_{i\rightarrow i+b}^{H}\big) \, X\ket{\phi_{i}^H}\big\|_2^2 \Bigr]
	$$
	where here and from now on, we use $X$ as a short hand for the projection $\proj{x}$ acting on~$\regX$. 
	The expectation is taken over $\Theta\in\mathcal Y$, $i\in\{0,...,q\}$ and $b\in\{0,1\}$ uniformly random. Note that the random bit $b \in \{0,1\}$ determines whether the measured query is answered with $H$ or with $H\!*\!\Theta x$. 
	
	We write ${\cal S}^{\cal A}[H]$ to emphasize that ${\cal S}$ only makes black-box access to ${\cal A}$ and that it depends on $H$. Our main technical lemma below then ensures that for any $H$ and for a random $\Theta \in \cal Y$, the success probability of ${\cal S}^{\cal A}[H]$ is up to an order-$q^2$ loss not much smaller than that of ${\cal A}^{H * \Theta x}$, and therefore not much smaller than that of ${\cal A}^{H}$ in case of a random $H$. 
	
	\begin{lem}
		\label{lem:mainresult}
		For any $H: {\cal X}\rightarrow {\cal Y}$ and  $x \in \cal X$, it holds that
		\begin{align*}
		\E_{\Theta,i,b}\left[\big\|G_x^{\Theta} \big(\mathcal{A}_{i+b\rightarrow q}^{H*\Theta x}\big)\big(\mathcal{A}_{i\rightarrow i+b}^{H}\big)X\ket{\phi_i^H}\big\|_2^2\right] 
		&\geq  \frac{\E_{\Theta}\Bigl[\big\|G_x^{\Theta} \ket{\phi_q^{H*\Theta x}}\big\|_2^2\Bigr]}{2(q+1)(2q+3) } - \frac{\big\|X\ket{\phi_q^H}\big\|_2^2}{2(q+1)|{\cal Y}| }.
		\end{align*}
		where the expectation is over random $\Theta \in {\cal Y}$, $i \in \{0,\ldots,q\}$ and $b \in \{0,1\}$.%
		\footnote{We consider $|{\cal Y}|$ to be superpolynomial in the security parameter, so that $\frac{1}{2(q+1)|{\cal Y}|}$ is negligible and can be neglected. In cases where $|{\cal Y}|$ is polynomial, the presented bound is not optimal, but an improved bound can be derived with the same kind of techniques.}
	\end{lem}
	\begin{proof}
		% Recall that $\ket{\phi_q^H} = \big(\mathcal{A}_{0\rightarrow q}^H\big)\ket{\phi_0}$. 
		We assume that the $\regY$-register of $\ket{\phi_q^H} = \big(\mathcal{A}_{0\rightarrow q}^H\big)\ket{\phi_0}$ is $\ket{0}$ no matter what $H$ is; this is without loss of generality since it can always be achieved by an insignificant modification to $\cal A$, i.e., by swapping $\regY$ with a default register within $\regE$. 
		For the purpose of the proof, we introduce an additional step ${\cal A}^H_{q \to q+1}$ that simply applies ${\cal O}^H$, and we expand the notions of $\ket{\phi_j^H}$ and ${\cal A}^H_{i \to j}$ to allow $j = q+1$. 
		%Understanding that this commutes with the measurement of register $\regY$, this simply computes the hash of the to-be-measured $x$ and writes it into register $\regW$. 
		Finally, we ``enhance'' $G_x^\Theta$ to%
		\footnote{Informally, these modifications mean that we let $\cal A$ make one more query to get $H(x)$ into register $\regY$, and $\tilde G_x^{H(x)}$ would then check that $\regY$ indeed contains $H(x)$. } 
		$$
		\tilde G_x^\Theta := G_x^\Theta (\mathbb{1} \otimes \proj{\Theta} \otimes \mathbb{1} \otimes \mathbb{1}) = X \otimes \proj{\Theta} \otimes \Pi_{x,\Theta} \otimes \mathbb{1} \, .
		$$
		%Note that $\|\tilde G_x^{H(x)} \ket{\phi_{q+1}^H}\|^2 = \|G_x^{H(x)} \ket{\phi_q^H}\|^2$. 
		For any $0\leq i \leq q$, inserting a resolution of the identity and exploiting that
		$$
		\big(\mathcal{A}_{i+1\rightarrow q+1}^{H*\Theta x}\big)\big(\mathcal{A}_{i\rightarrow i+1}^H\big)\big(\mathbb{1}-X\big)\ket{\phi_i^H} = \big(\mathcal{A}_{i\rightarrow q+1}^{H*\Theta x}\big)\big(\mathbb{1}-X\big)\ket{\phi_i^H} \, ,
		$$ 
		we can write
		\switch{
			\begin{align*}
			\big(\mathcal{A}_{i+1\rightarrow q+1}^{H*\Theta x}\big)\ket{\phi_{i+1}^H} 
			&=  \big(\mathcal{A}_{i+1\rightarrow q+1}^{H*\Theta x}\big)\big(\mathcal{A}_{i\rightarrow i+1}^H\big)\big(\mathbb{1}-X\big)\ket{\phi_i^H} \hspace{-8ex}\!\!&\!\! +\: \big(\mathcal{A}_{i+1\rightarrow q+1}^{H*\Theta x}\big)\big(\mathcal{A}_{i\rightarrow i+1}^H\big)X\ket{\phi_i^H} \notag\\
			&= \big(\mathcal{A}_{i\rightarrow q+1}^{H*\Theta x}\big)\big(\mathbb{1}-X\big)\ket{\phi_i^H} \hspace{-8ex}\!\!&\!\! +\: \big(\mathcal{A}_{i+1\rightarrow q+1}^{H*\Theta x}\big)\big(\mathcal{A}_{i\rightarrow i+1}^H\big)X\ket{\phi_i^H}\\
			&= \big(\mathcal{A}_{i\rightarrow q+1}^{H*\Theta x}\big)\ket{\phi_i^H} - \big(\mathcal{A}_{i\rightarrow q+1}^{H*\Theta x}\big)X\ket{\phi_i^H} \hspace{-8ex}\!\!&\!\! +\: \big(\mathcal{A}_{i+1\rightarrow q+1}^{H*\Theta x}\big)\big(\mathcal{A}_{i\rightarrow i+1}^H\big)X\ket{\phi_i^H}
			\end{align*}
		}{
			\begin{align*}
			&\big(\mathcal{A}_{i+1\rightarrow q+1}^{H*\Theta x}\big)\ket{\phi_{i+1}^H} &\\
			&\hspace{30pt}=  \big(\mathcal{A}_{i+1\rightarrow q+1}^{H*\Theta x}\big)\big(\mathcal{A}_{i\rightarrow i+1}^H\big)\big(\mathbb{1}-X\big)\ket{\phi_i^H} \!\!&\!\! +\: \big(\mathcal{A}_{i+1\rightarrow q+1}^{H*\Theta x}\big)\big(\mathcal{A}_{i\rightarrow i+1}^H\big)X\ket{\phi_i^H} \notag\\
			&\hspace{30pt}= \big(\mathcal{A}_{i\rightarrow q+1}^{H*\Theta x}\big)\big(\mathbb{1}-X\big)\ket{\phi_i^H} \!\!&\!\! +\: \big(\mathcal{A}_{i+1\rightarrow q+1}^{H*\Theta x}\big)\big(\mathcal{A}_{i\rightarrow i+1}^H\big)X\ket{\phi_i^H}\\
			&\hspace{30pt}= \big(\mathcal{A}_{i\rightarrow q+1}^{H*\Theta x}\big)\ket{\phi_i^H} - \big(\mathcal{A}_{i\rightarrow q+1}^{H*\Theta x}\big)X\ket{\phi_i^H} \!\!&\!\! +\: \big(\mathcal{A}_{i+1\rightarrow q+1}^{H*\Theta x}\big)\big(\mathcal{A}_{i\rightarrow i+1}^H\big)X\ket{\phi_i^H}
			\end{align*}
		}
		Rearranging terms, applying $\tilde G_{x}^\Theta$ and using the triangle equality, we can thus bound
		\begin{align*}
		\big\| \tilde G_{x}^\Theta  \big(\mathcal{A}_{i\rightarrow q+1}^{H*\Theta x}\big) \ket{\phi_i^H} \big\|_2 
		\leq \big\|\tilde G_{x}^\Theta & \big(\mathcal{A}_{i+1\rightarrow q+1}^{H*\Theta x}\big)\ket{\phi_{i+1}^H}\big\|_2 \\& + \big\| \tilde G_{x}^\Theta \big(\mathcal{A}_{i\rightarrow q+1}^{H*\Theta x}\big)X\ket{\phi_i^H}\big\|_2\\ & \qquad + \big\| \tilde G_{x}^\Theta \big(\mathcal{A}_{i+1\rightarrow q+1}^{H*\Theta x}\big)\big(\mathcal{A}_{i\rightarrow i+1}^H\big)X\ket{\phi_i^H}\big\|_2 \, .
		\end{align*}
		Summing up the respective sides of the inequality over $i=0,\ldots,q$, we get
		\begin{equation*}
		\big\|\tilde G_x^\Theta\ket{\phi_{q+1}^{H*\Theta x}}\big\|_2 \:\leq\: \big\| \tilde G_x^\Theta\ket{\phi_{q+1}^H}\big\|_2 + \!\!\!\sum_{0\leq i \leq q \atop b\in \{0,1\}}\!\!\! \big\|\tilde G_x^\Theta \big(\mathcal{A}_{i+b\rightarrow q+1}^{H*\Theta x}\big)\big(\mathcal{A}_{i\rightarrow i+b}^H\big)X\ket{\phi_i^H}\big\|_2 \, .\label{eqnsquarthis}
		\end{equation*}
		By squaring both sides, dividing by $2q+3$ (i.e., the number of terms on the right hand side), and using Jensen's inequality on the right hand side, we obtain 
		$$
		\frac{\big\|\tilde G_x^\Theta\ket{\phi_{q+1}^{H*\Theta x}}\big\|_2^2}{2q+3} \leq \big\|\tilde G_x^\Theta\ket{\phi_{q+1}^H}\big\|_2^2 + \!\!\!\sum_{0\leq i \leq q \atop b\in \{0,1\}}\!\!\!\big\| \tilde G_x^\Theta \big(\mathcal{A}_{i+b\rightarrow q+1}^{H*\Theta x}\big)\big(\mathcal{A}_{i\rightarrow i+b}^H\big)X\ket{\phi_i^H}\big\|_2^2  
		$$
		% 		\begin{align*}
		% 		&(2q+3)\cdot\left(\| \tilde G_y^\Theta\ket{\phi_{q+1}^H}\|_2^2 + \sum_{\substack{0\leq i \leq q\\b\in \{0,1\}}}\big\| \tilde G_y^\Theta \big(\mathcal{A}_{i+b\rightarrow q+1}^{H*\Theta y}\big)\big(\mathcal{A}_{i\rightarrow i+b}^H\big)X\ket{\phi_i^H}\big\|_2^2\right) \\
		% 		&\hspace{250pt}\geq\hspace{8pt} \|\tilde G_y^\Theta\ket{\phi_{q+1}^{H*\Theta y}}\|_2^2 \, .
		% 		\end{align*}
		and thus
		\begin{equation}\label{eq:intermediate0}
		\E_{i,b}\left[\big\|\tilde G_x^\Theta \big(\mathcal{A}_{i+b\rightarrow q+1}^{H*\Theta x}\big)\big(\mathcal{A}_{i\rightarrow i+b}^H\big)X\ket{\phi_i^H}\big\|_2^2\right] \;\geq\; \frac{\big\|\tilde G_x^\Theta\ket{\phi_{q+1}^{H*\Theta x}}\big\|_2^2}{2(q+1)(2q+3)} - \frac{\big\|\tilde G_x^\Theta\ket{\phi_{q+1}^H}\big\|_2^2}{2(q+1)} \, .
		\end{equation}
		Since both $\proj{\Theta}$ and $\mathcal{A}_{q\rightarrow q+1}^{H*\Theta x} = {\cal O}^{H*\Theta x}$ commute with $G_x^\Theta$, we get 
		\begin{align}
		\big\|\tilde G_x^\Theta \big(\mathcal{A}_{i+b\rightarrow q+1}^{H*\Theta x}\big)\big(\mathcal{A}_{i\rightarrow i+b}^H\big)X\ket{\phi_i^H}\big\|_2^2 \,&\leq 
		\big\|G_x^\Theta \big(\mathcal{A}_{i+b\rightarrow q+1}^{H*\Theta x}\big)\big(\mathcal{A}_{i\rightarrow i+b}^H\big)X\ket{\phi_i^H}\big\|_2^2 \nonumber \\ 
		&= 
		\big\|G_x^\Theta \big(\mathcal{A}_{i+b\rightarrow q}^{H*\Theta x}\big)\big(\mathcal{A}_{i\rightarrow i+b}^H\big)X\ket{\phi_i^H}\big\|_2^2 \, .\label{eq:intermediate1}
		\end{align}
		Also, because $(X \otimes \proj{\Theta})  {\cal O}^{H*\Theta x} = (X \otimes \mathbb{1}) {\cal O}^{H*\Theta x}$, and ${\cal O}^{H*\Theta x}$ commutes with $G_x^\Theta$, we get
		\begin{equation}\label{eq:intermediate2}
		\bigl\|\tilde G_x^\Theta\ket{\phi_{q+1}^{H*\Theta x}}\bigr\|_2^2 = \bigl\| G_x^\Theta\ket{\phi_{q}^{H*\Theta x}}\bigr\|_2^2 \, .
		\end{equation}
		Finally,
		\begin{equation}\label{eq:intermediate3}
		\E_{\Theta}\left[\bigl\|\tilde G_x^\Theta\ket{\phi_{q+1}^H}\bigr\|_2^2\right] \leq \E_{\Theta}\left[\bigl\| (X \otimes \proj{\Theta}) {\cal O}^H \ket{\phi_q^H}\bigr\|_2^2\right] \leq \frac{1}{|{\cal Y}|} \bigl\| X \ket{\phi_q^H}\bigr\|_2^2\, .
		\end{equation}
		Inserting \eqref{eq:intermediate1}, \eqref{eq:intermediate2} and \eqref{eq:intermediate3} into  \eqref{eq:intermediate0} yields the claimed result.
		\qed
	\end{proof}

	\subsection{Switching notation, and simulating the random oracle}
	Introducing more algorithmic-probabilistic notation, we write
	$$
	(x,x',z) \leftarrow \langle{\cal S}^{\cal A}[H] , \Theta\rangle
	$$
	to specify the probability space determined as follows, relying on the above construction of the two-stage algorithm $\cal S$ when given $\cal A$. In the first stage ${\cal S}^{\cal A}[H]$ produces $x$, and then in the second stage, upon receiving $\Theta$, it produces $x'$ and~$z$, where $z$ may be quantum. Our figure of merit above, i.e., the left hand side of the bound in Lemma~\ref{lem:mainresult} (with $x$ replaced by $x_\circ$), is then denoted by 
	$$
	\Pr_\Theta\bigr[x\!=\! x_\circ \wedge x'\!=\! x_\circ \wedge V(x,\Theta,z) : (x,x',z) \leftarrow \langle{\cal S}^{\cal A}[H] , \Theta\rangle\bigl]  \, ,
	$$
	where the subscript $\Theta$ in $\Pr_\Theta$ denotes that the probability is averaged over a random choice of $\Theta$.
	
	Using this notation, but also weakening the bound slightly by not requiring $x' = x_\circ$, for any $H$ and $x_\circ$ the bound from Lemma~\ref{lem:mainresult} then becomes 
	\begin{align*}
	\Pr_\Theta\bigr[x\!=\! x_\circ& \wedge V(x,\Theta,z) : (x,z) \leftarrow \langle{\cal S}^{\cal A}[H] , \Theta\rangle\bigl] 
	\switch{\;}{\\&}
	\gtrsim \frac{1}{O(q^2)} \Pr_\Theta\bigr[x\!=\! x_\circ \wedge V(x,H(x),z) : (x,z) \leftarrow {\cal A}^{H*\Theta x} \bigl]
	\end{align*}
	\sloppy where the approximate inequality $\gtrsim$ hides the term
	\switch{
		$\frac{1}{ 2(q+1)|{\cal Y}|}\Pr_H\bigl[x\!=\! x_\circ : (x,z) \leftarrow {\cal A}^{H}\bigr].$
	}
	{
		$$\frac{1}{ 2(q+1)|{\cal Y}|}\Pr_H\bigl[x\!=\! x_\circ : (x,z) \leftarrow {\cal A}^{H}\bigr].$$
	}  
	%which is bounded by $\frac{1}{ 2q|{\cal Y}|}$ when summed over all $x_\circ$. 
	Recall that the output $z$ may be a quantum state, in which case the predicate $V$ is given by a measurement that depends on $x$, and $H(x)$ or $\Theta$,  respectively. 
	
	We fix a family $\cal H$ of $2(q+1)$-wise independent hash functions and average the above inequality over a random choice of $H \in \cal H$ from this family. We simply write $\cal S$ for ${\cal S}[H]$ with $H$ chosen like that. Furthermore, we observe that, for any fixed $x$, the family $\{H *\Theta x \,|\, H \in {\cal H}, \Theta \in \{0,1\}^n\}$ is a family of $2(q+1)$-wise independent hash functions as well. Finally, we use that $\cal A$ (together with the check $V(x,H(x),z)$) cannot distinguish a random function $H\!*\!
	\Theta x$ in that family from a fully random function $H$ \cite{Zhandry2012a}.  This gives us the following variation of Lemma~\ref{lem:mainresult}, which we state as our main technical theorem:

	\begin{thm}[Measure-and-reprogram]\label{thmmain} 
		Let ${\cal X},{\cal Y}$ be finite non-empty sets. 
		There exists a black-box polynomial-time two-stage quantum algorithm $\cal S$ with the following property. 
		Let $\cal A$ be an arbitrary oracle quantum algorithm that makes $q$ queries to a uniformly random $H: {\cal X}\rightarrow {\cal Y}$ and that outputs some $x \in {\cal X}$ and a (possibly quantum) output~$z$. Then, the two-stage algorithm ${\cal S}^{\cal A}$ outputs some $x \in {\cal X}$ in the first stage and, upon a random $\Theta \in {\cal Y}$ as input to the second stage, a (possibly quantum) output $z$, so that for any \mbox{$x_\circ\in {\cal X}$} and any predicate%
		\footnote{We recall that in case $z$ is a quantum state, $V$ is given by means of a measurement. }
		$V$:
		\begin{align*}
		\Pr_\Theta\bigr[x\!=\! x_\circ& \wedge V(x,\Theta,z) : (x,z) \leftarrow \langle{\cal S}^{\cal A} , \Theta\rangle\bigl] 
		\switch{\;}{\\ &}
		\gtrsim \frac{1}{O(q^2)} \Pr_H\bigl[x\!=\! x_\circ \wedge V(x,H(x),z) : (x,z) \leftarrow {\cal A}^{H} \bigr] \, ,
		\end{align*}
		where the $\gtrsim$ hides a term that is bounded by $\frac{1}{2q|{\cal Y}|}$ when summed over all $x_\circ$.\footnote{Note added: In follow-up work, \cite{DFM20} proves a slightly improved version of this theorem that avoids the (negligible) additive error term. As a consequence, the additive error term in Theorem \ref{thmFS} can be avoided as well.}% 
	\end{thm}

	\begin{rem}\label{rem:BB}
		We do not spell out in detail what it means for a quantum algorithm like $\cal S$ to be {\em black-box}; see e.g.~\cite{Unruh2017} for a rigorous definition. What we obviously need here is that ${\cal S}^{\cal A}$ has access to $\cal A$'s initial state $\ket{\phi_0}$ and to $q$, and is given black-box access to the unitaries $A_i$. Furthermore, for later purposes, we need the following composition property: if ${\cal S}$ is a black-box algorithm with access to ${\cal A}$, and ${\cal K}$ is a black-box algorithm with access to ${\cal S}^{\cal A}$, then there exists a black-box algorithm ${\cal K}^{\cal S}$ with access to $\cal A$ so that \smash{$({\cal K}^{\cal S})^{\cal A} = {\cal K}^{({\cal S}^{\cal A})}$}. 
	\end{rem}

	\section{Security of the Fiat-Shamir Transformation}

	In this section, we show how to reduce security of the Fiat-Shamir transformation to the security of the underlying \sigp: any dishonest prover attacking the Fiat-Shamir transformation can be turned into a dishonest prover that succeeds to break the underlying \sigp with the same probability up to a polynomial loss. This reduction is obtained by a straightforward application of Theorem~\ref{thmmain}. 
	Our security reduction holds very generically and is not strongly tight to the considered notion of security, as long as the respective security definitions for the \sigp and the Fiat-Shamir transformation ``match up''.

	\subsection{\sigps}
	
	We recall the definition of a \sigp. 
	\begin{defi}[\sigp]
		A \sigp $\mathsf{\Sigma} = ({\cal P}, {\cal V})$ for a 
		% language $\mathcal{L}_\eta\in \mathsf{NP}$ with security parameter $\eta$ and witness 
		relation $R \subseteq {\cal X} \times {\cal W}$ is a three-round two-party interactive protocol of the form:
		\begin{empheq}[box=\widefbox]{align*}
		&\underline{\text{Prover } {\cal P}(x,w)}&&&&\underline{\text{Verifier } {\cal V}(x)}\\
		&&&\overset{a}{\longrightarrow}&\\
		&&&\overset{c}{\longleftarrow}&&c\overset{\,\$}{\leftarrow} {\cal C} \\
		&&&\overset{z}{\longrightarrow}&&\textup{Accept iff } V(x,a,c,z) = 1
		\end{empheq}\switch{\vspace{-1.5ex}}{}
	\end{defi}
	Using our terminology and notation from above, $\cal P$ is a two-stage algorithm and we can write
	$$
	(a,z) \leftarrow \langle {\cal P}(x,w) , c\rangle
	$$ 
	for the generation of the first message $a$ in the first stage and the reply $z$ in the second stage once given the challenge $c$. 
	
	\begin{rem}
		We allow the set of {\em instances} $\cal X$, the set of {\em witnesses} $\cal W$ and the relation $R$ to depend on a security parameter $\eta$. Similarly, the interactive algorithms ${\cal P}$ and ${\cal V}$ may depend on $\eta$ (or have $\eta$ as part of their input). However, for ease of notation, we suppress these dependencies on $\eta$ unless they are crucial. 
	\end{rem}
	
	\begin{rem}
		We do not necessarily require a \sigp to be perfectly or statistically correct. This allows us to include protocols that use {\em rejection sampling}, where with a constant probability, the value $z$ would leak too much information on the witness $w$ and so the prover sends $\bot$ instead. 
		On the other hand, by default we consider the soundness/knowledge error to be negligible, i.e., a dishonest prover succeeds only with negligible probability to make the verifier accept if $x$ is not a valid instance or the prover has no witness for it (depending on the considered soundness notion). Negligible soundness/knowledge error can always be achieved by parallel repetition (see e.g. \cite{Damgard2010}).
	\end{rem}

	\subsection{The Fiat-Shamir transformation} \label{sec:FStrans}
	The {\em Fiat-Shamir transformation} turns a \sigp $\mathsf{\Sigma}$ into a non-interactive proof system, denoted $\mathsf{FS[\Sigma]}$, by replacing the verifier's random choice of $c \in \cal C$ with $c := H(x,a)$, where $H : {\cal X}' \to {\cal C}$ is a hash function with a domain ${\cal X}'$ that contains all pairs $x' = (x,a)$ with $x \in \cal X$ and $a$ produced by $\cal P$. 
	In other words, upon input $x$ and $w$, the honest FS-prover produces $\pi = (a,z)$ by running the two-stage \sigp prover $\cal P$ but using $c = H(x,a)$ as challenge (i.e., as input to the second stage).  
	In case $\mathsf{\Sigma}$ is not statistically correct, the above process of producing $\pi = (a,z)$ is repeated sufficiently many times until $V(x,a,H(x,a),z)$ is satisfied (or some bound is reached). 
	In either case, we will write this as
	$$
	\pi = (a,z) \leftarrow P_{FS}^H(x,w) \, . 
	$$
	We may write as $V_{FS}^H(x,\pi)$ the FS-verifier's check whether $V(x,a,H(x,a),z)$ is satisfied or not.
	In the security analysis, the hash function $H$ is modeled by a random oracle, i.e. by oracle access to a uniformly random $H : {\cal X}' \to {\cal C}$. 
	% ; thus, the considered algorithms are oracle algorithms that make oracles calls to a uniformly random function $H$ with appropriate domain and range. 
	
	When considering an {\em adversary} $\cal A$ that tries to {\em forge} a proof for some instance~$x \in \cal X$, one can distinguish between an {\em arbitrary but fixed} $x$, and an $x$ that is {\em chosen} by $\cal A$ and output along with $a$ in case of \sigps, respectively along with $\pi$ in case of the Fiat-Shamir transformation. If $x$ is fixed then the adversary is called {\em static}, otherwise it is called {\em adaptive}. For the typical security definitions for \sigps this distinction between a static and an adaptive $\cal A$ makes no difference (see Lemmas~\ref{lemma:StatVsAaptSound} and \ref{lemma:StatVsAaptPoK} below), but for the Fiat-Shamir transformation it (potentially) does.

	\subsection{The generic security reduction}\label{SecFSred}
	
	Since an adaptive adversary is clearly not less powerful than a static adversary, we restrict our attention for the moment to the adaptive case. Recall that such an adaptive FS-adversary $\cal A$ outputs the instance $x \in \cal X$ along with the proof $\pi = (a,z)$, and the figure of merit is the probability that $x,a,z$ satisfies $V(x,a,H(x,a),z)$. Thus, we can simply 
	%Very much along the same lines, in case of an {\em adaptive} FS-adversary ${\cal A}$ that outputs $x$ along with a proof $\pi = (a,z)$ for $x$, we can 
	apply Theorem \ref{thmmain}, with $(x,a)$ playing the role of what is referred to as $x$ in the theorem statement, to obtain the existence of an adaptive $\Sigma$-adversary ${\cal S}^{\cal A}$ that produces $(x,a)$ in a first stage, and upon receiving a random challenge $c$ produces $z$, such that for any $x_\circ \in {\cal X}$
	\begin{align*}
	\Pr_{c}\bigr[x\!=\! x_\circ& \wedge V(x,a,c,z) : (x,a,z) \leftarrow \langle{\cal S}^{\cal A} , c\rangle\bigl] 
	\switch{\;}{\\ &}
	\gtrsim \frac{1}{O(q^2)} \Pr_H\bigr[x\!=\! x_\circ \wedge V(x,a,H(x,a),z) : (x,a,z) \leftarrow {\cal A}^H \bigl] \, ,
	\end{align*}
	where the approximate inequality hides a term that is bounded by $\frac{1}{ 2q|{\cal C}|}$ when summed over all $x_\circ \in {\cal X}$. Understanding that $x$ is given to $\cal V$ along with the first message $a$ but also treating it as an output of ${\cal S}^{\cal A}$, while $\cal V$'s output $v$ is its decision to accept or not, we write this as 
	% 	We observe that the probability on the left hand side equals the probability of ${\cal S}^{\cal A}$ being successful in making the \sigp verifier ${\cal V}$ accept for the particular statement $x_\circ$, and the probability on the right hand side equals the probability of the FS-adversary ${\cal A}$ being successful in producing a proof for the particular statement $x_\circ$. It is thus meaningful to write the above as \serge{Not fully sure about the notation. }
	\begin{align*}
	\Pr\bigr[x\!=\! x_\circ \wedge v = accept& :(x,v) \leftarrow \langle{\cal S}^{\cal A} , {\cal V}\rangle\bigl] 
	\switch{\;}{\\ &}
	\gtrsim \frac{1}{O(q^2)} \Pr_H\bigr[x\!=\! x_\circ \wedge V^H_{FS}(x,\pi) : (x,\pi) \leftarrow {\cal A}^H \bigl] \, .
	\end{align*}
	Summed over all $x_\circ \in {\cal X}$, this in particular implies that
	\begin{align*}
	\Pr\bigr[\langle{\cal S}^{\cal A} , {\cal V}\rangle = accept \bigl] \; \geq \frac{1}{O(q^2)} \Pr_H\bigr[V^H_{FS}(x,\pi) : (x,\pi) \leftarrow {\cal A}^H \bigl] - \frac{1}{2q|{\cal C}|} \, .
	\end{align*}
	
	\begin{rem}\label{rem:Generalization}
		We point out that the above arguments extend to a FS-adversary ${\cal A}$ that, besides the instance $x$ and the proof $\pi = (a,z)$, also produces some local (possibly quantum) output satisfying some (quantum) predicate that may depend on $x,a,z$. The resulting $\Sigma$-adversary ${\cal S}^{\cal A}$ is then ensured to produce a local output that satisfies the considered predicate as well, up to the given loss in the probability. Indeed, we can simply include this local output in $z$ and extend the predicate $V$ accordingly. 
	\end{rem}
	
	In a very broad sense, the above means that for {\em any} FS-adversary ${\cal A}$ there exists a $\Sigma$-adversary ${\cal S}^{\cal A}$ that ``{\em achieves the same thing}'' up to a $O(q^2)$ loss in success probability. Hence, for matching corresponding security definitions, security of a \sigp (against a dishonest prover) implies security of its Fiat-Shamir transform. % We will exemplify this result in the upcoming sections for the standard notions of (knowledge) soundness. 
	
	%We highlight here the above transformation from an adaptive FS-adversary ${\cal A}$ to the adaptive $\Sigma$-adversary ${\cal S}^{\cal A}$. 
	We summarize here the above basic transformation from an adaptive FS-adversary ${\cal A}$ to an adaptive $\Sigma$-adversary ${\cal S}^{\cal A}$. 
	\begin{thm}\label{thmFS} 
		There exists a black-box quantum polynomial-time two-stage quantum algorithm $\cal S$ such that for any adaptive Fiat-Shamir adversary $\cal A$, making $q$ queries to a uniformly random function $H$ with appropriate domain and range, and for any $x_\circ \in {\cal X}$: 
		\begin{align*}
		\Pr\bigr[x\!=\! x_\circ \wedge v = accept& :(x,v) \leftarrow \langle{\cal S}^{\cal A} , {\cal V}\rangle\bigl] 
		\switch{\;}{\\ &}
		\gtrsim \frac{1}{O(q^2)} \Pr_H\bigr[x\!=\! x_\circ \wedge V^H_{FS}(x,\pi) : (x,\pi) \leftarrow {\cal A}^H \bigl] \, ,
		\end{align*}
		where the $\gtrsim$ hides a term that is bounded by $\frac{1}{2q|{\cal C}|}$ when summed over all $x_\circ$. 
	\end{thm}
	Below, we apply the above general reduction to the respective standard definitions for \textit{soundness} and \textit{proof of knowledge}. Each property comes in the variants \textit{computational} and \textit{statistical}, for guarantees against computationally bounded or unbounded adversaries respectively, and one may consider the static or the adaptive case.

	\subsection{Preservation of soundness}\label{subsec:soundpres}
	
	Let $\mathsf{\Sigma} = ({\cal P},{\cal V})$ be a \sigp for a relation $R$, and let $\mathsf{FS[\Sigma]}$ be its Fiat-Shamir transformation. We set $\mathcal{L} := \{x \in {\cal X}\,|\, \exists\, w \in {\cal W} : R(x,w)\}$. It is understood that ${\cal P}$ and ${\cal V}$, as well as $R$ and thus $\cal L$, may depend on a security parameter $\eta$. 
	% We recall again that an {\em adaptive} adversary produces the instance $x$ for which it tries to forge a proof, whereas a ordinary (i.e. {\em static}) adversary has to work for an arbitrary but {\em fixed} $x$.
	We note that in the following definition, we overload notation a bit by writing $\cal A$ for both for the ordinary static and for the adaptive adversary (even though a given $\cal A$ is usually either static or adaptive).

	\begin{defi}\label{def:soundness}
		$\mathsf{\Sigma}$ is \textbf{(computationally/statistically) sound} if there exists a negligible function $\mu(\eta)$ such that for any (quantum polynomial-time/unbounded) adversary $\mathcal{A}$ and any $\eta \in \mathbb{N}$: 
		$$
		\Pr\left[\langle{\cal A} , {\cal V}(x)\rangle = accept \right] \leq \mu(\eta) 
		$$ 
		for all $x\notin \mathcal{L}$; respectively, in case of an {\bf adaptive} $\mathcal{A}$:  
		$$
		\Pr\left[x \not\in {\cal L} \,\wedge\, v = accept :(x,v) \leftarrow \langle{\cal A} , {\cal V}\rangle \right] \leq \mu(\eta) \, .
		$$ 
		$\mathsf{FS[\Sigma]}$ is \textbf{(computationally/statistically) sound} if there exists a negligible function $\mu(\eta)$ and a constant $e$ such that for any (quantum polynomial-time/unbounded) adversary $\mathcal{A}$ and any $\eta \in \mathbb{N}$: 
		$$
		\Pr_H\left[V_{FS}^H(x,\pi) : \pi \leftarrow \mathcal{A}^H\right] \leq q^e \mu(\eta) 
		$$
		for all $x\notin \mathcal{L}$; respectively, in case of an {\bf adaptive} $\mathcal{A}$:  
		$$
		\Pr_H\left[V_{FS}^H(x,\pi) \wedge x\notin \mathcal{L} : (x,\pi)\leftarrow \mathcal{A}^H\right] \leq q^e \mu(\eta) \, .
		$$
	\end{defi}
	
	\begin{rem}
		Note that for the soundness of $\mathsf{FS[\Sigma]}$, the adversary $\cal A$'s success probability may unavoidably grow with the number $q$ of oracle queries, but we require that it grows only polynomially in $q$. 
	\end{rem}
	
	\begin{rem}
		In line with Section~\ref{secRepr}, the description of a quantum algorithm $\cal A$ is understood to include the initial state $\ket{\phi_0}$. As such, when quantifying over all $\cal A$ it is understood that this includes a quantification over all $\ket{\phi_0}$ as well. This stays true when considering $\cal A$ to be quantum polynomial-time, which means that the unitaries $A_i$ can be computed by polynomial-time quantum circuits, and $q$ is polynomial in size, but does not put any restriction on $\ket{\phi_0}$.%
		\footnote{In other words, $\cal A$ is then {\em non-uniform} quantum polynomial-time with {\em quantum} advice. }
		This is in line with \cite[Def.~1]{Unruh2012}, which explicitly spells out this quantification. 
	\end{rem}
	
	We consider the following to be folklore knowledge; for completeness, we still give a proof in Appendix~\ref{app:proofs}.
	\begin{lem}\label{lemma:StatVsAaptSound}
		If $\mathsf{\Sigma}$ is computationally/statistically sound for {\em static} adversaries then it is also computationally/statistically sound for {\em adaptive} adversaries. 
	\end{lem}
	The following is now an immediate application of Theorem~\ref{thmFS} and the above observation regarding static and adaptive security for \sigps. 
	
	\begin{cor}\label{corProof}
		Let $\mathsf{\Sigma}$ be a \sigp with superpolynomially sized challenge space $\mathcal{C}$. If $\mathsf{\Sigma}$ is computationally/statistically sound against a static adversary then  $\mathsf{FS[\Sigma]}$ is computationally/statistically sound against an adaptive adversary. 
	\end{cor}
	\begin{proof}
		Applying Theorem \ref{thmFS}, we find that for any adaptive FS-adversary $\cal A$, polynomially bounded in the computational setting, there exists an adaptive \sigp adversary $\mathcal{S}^\mathcal{A}$, polynomially bounded if $\cal A$ is, so that 
		\begin{align*}
		\Pr\bigl[&x\notin \mathcal{L}  \,\wedge\, V_{FS}^H(x,\pi):  (x,\pi)\leftarrow \mathcal{A}^H\bigr] \\[1ex]
		&= \sum_{x_\circ\notin \mathcal{L}}\Pr\bigl[x\!=\!x_\circ  \,\wedge\, V_{FS}^H(x,\pi):  (x,\pi)\leftarrow \mathcal{A}^H\bigr] \\
		% 		&\hspace{5pt}=\hspace{5pt} \sum_{x\notin \mathcal{L}_\eta}\Pr\left[\mathcal{A}^H\to (x,\pi^{ok})\right]\\
		&\leq O(q^2)\cdot\bigg(\bigg(\sum_{x_\circ\notin \mathcal{L}}\Pr\bigr[x\!=\! x_\circ \wedge v = accept :(x,v) \leftarrow \langle{\cal S}^{\cal A} , {\cal V}\rangle\bigl]\bigg) + \frac{1}{2q|\mathcal{C}|}\bigg) \\
		&= O(q^2)\cdot\bigg(\Pr\bigr[x\not\in {\cal L} \wedge v = accept :(x,v) \leftarrow \langle{\cal S}^{\cal A} , {\cal V}\rangle\bigl] \bigg) + \frac{O(q)}{|\mathcal{C}|}\\
		&\leq O(q^2)\cdot\mu(\eta) + \frac{O(q)}{|\mathcal{C}|}\hspace{20pt} 
		\end{align*}
		where the last inequality holds for some negligible function $\mu(\eta)$ if $\mathsf{\Sigma}$ is sound against an adaptive adversary. The latter is ensured by the assumed soundness against a static adversary and Lemma~\ref{lemma:StatVsAaptSound}. This bound can obviously be written as $q^2 \mu'(\eta)$ for another negligible function $\mu'(\eta)$, showing the claimed soundness of $\mathsf{FS[\Sigma]}$. \qed
	\end{proof}

	\subsection{Preservation as a proof of knowledge}\label{subsec:extractpres}
	
	We now recall the definition of a proof of knowledge, sometimes also referred to as (witness) extractability, tailored to the case of a negligible ``knowledge error''. 
	Informally, the requirement is that if $\cal A$ succeeds in proving an instance $x$, then by using $\cal A$ as a black-box only it is possible to extract a witness for $x$. In case of an arbitrary but fixed $x$, this property is formalized in a rather straightforward way; however, in case of an adaptive $\cal A$, the formalization is somewhat subtle, because one can then not refer to {\em the} $x$ for which $\cal A$ manages to produce a proof. We adopt the approach (though not the precise formalization) from~\cite{Unruh2017}, which requires $x$ to satisfy an arbitrary but fixed predicate. 
	
	\begin{defi}\label{def:PoK}
		$\mathsf{\Sigma}$ is a \textbf{(computational/statistical) proof of knowledge} if there exists a quantum polynomial-time 
		black-box `knowledge extractor' $\mathcal{K}$, a polynomial $p(\eta)$, a constant $d \geq 0$, and a negligible
		% \footnote{While in general this property may be defined for a non-negligible knowledge error, in the context of Fiat-Shamir we require $\kappa(\eta)$ to be negligible. Note that a negligible knowledge error can always be obtained by sequentially repeating the protocol.} 
		function $\kappa(\eta)$ such that for any (quantum polynomial-time/unbounded) adversary $\mathcal{A}$, any $\eta \in \mathbb{N}$ and any $x \in \cal X$: 
		\begin{align*}
		&\Pr\left[(x,w)\in R : w\leftarrow \mathcal{K}^{\mathcal{A}}(x)\right] \geq\frac{1}{p(\eta)}\cdot\Pr\left[\langle{\cal A}, {\cal V}(x)\rangle = accept \right]^d- \kappa(\eta) \, ;
		\end{align*}
		respectively, in case of an {\bf adaptive} $\cal A$:  
		\begin{align*}
		\Pr\bigl[x\in X &\,\wedge\, (x,w)\in R: (x,w)\leftarrow \mathcal{K}^{\mathcal{A}}\bigr] 
		\switch{\;}{\\[0.4em] &}
		\geq \frac{1}{p(\eta)} \cdot \Pr\left[ x\in X \wedge v = accept: (x,v)\leftarrow \langle{\cal A}, {\cal V}\rangle\right]^d - \kappa(\eta)
		\end{align*}
		for any subset $X\subseteq {\cal X}$. 
		
		$\mathsf{FS[\Sigma]}$ is a \textbf{(computational/statistical) proof of knowledge} if there exists a polynomial-time black-box `knowledge extractor' $\mathcal{E}$, a polynomial $p(\eta)$, constants $d,e \geq 0$, and a negligible function $\mu(\eta)$, such that for any (quantum polynomial-time/unbounded) algorithm $\mathcal{A}$, any $\eta \in \mathbb{N}$ and any $x \in \cal X$: 
		\begin{align*}
		\Pr\bigl[(x,w)\in R: w\leftarrow \mathcal{E}^{\mathcal{A}}(x)\bigr] 
		&\geq \frac{1}{q^e p(\eta)}\cdot\Pr_H\left[V_{FS}^H(x,\pi): \pi \leftarrow \mathcal{A}^H\right]^d - \mu(\eta). \; ;
		\end{align*}
		respectively, in case of an {\bf adaptive} $\cal A$: 
		% uniformly random $H:\{0,1\}^{\ell_x+\ell_{com}}\rightarrow \{0,1\}^{\ell_{ch}}$ and 
		\begin{align*}
		\Pr\bigl[x\in X &\,\wedge\, (x,w)\in R: (x,w)\leftarrow \mathcal{E}^{\mathcal{A}}\bigr] 
		\switch{\;}{\\[0.4em] &}
		\geq \frac{1}{q^e p(\eta)}\cdot\Pr_H\left[ x\in X \wedge V_{FS}^H(x,\pi): (x,\pi)\leftarrow \mathcal{A}^H\right]^d - \mu(\eta) 
		\end{align*}
		for any subset $X\subseteq {\cal X}$, where $q$ is the number of queries $\cal A$ makes. 
	\end{defi}
	Also here, for \sigps static security implies adaptive security. 
	
	\begin{lem}\label{lemma:StatVsAaptPoK}
		If $\mathsf{\Sigma}$ is a computational/statistical proof of knowledge for {\em static} $\cal A$ then it is also a computational/statistical proof of knowledge for {\em adaptive} $\cal A$. 
	\end{lem}
	Again, the following is now an immediate application of Theorem~\ref{thmFS} and the above observation regarding static and adaptive security for \sigps. 
	
	\begin{cor}\label{corPoK}
		Let $\mathsf{\Sigma}$ be a \sigp with superpolynomially sized $\mathcal{C}$. If $\mathsf{\Sigma}$ is a computational/statistical proof of knowledge for static adversaries then $\mathsf{FS[\Sigma]}$ is a computational/statistical proof of knowledge for adaptive adversaries. 
	\end{cor}
	
	\begin{proof}
		First, we observe that by Lemma~\ref{lemma:StatVsAaptPoK}, we may assume $\mathsf{\Sigma}$ to be a computational/statistical proof of knowledge for {\em adaptive} adversaries. Let $\cal K$ be the black-box knowledge extractor. Let $\cal A$ be an (quantum polynomial-time/unbounded) adaptive FS-adversary $\cal A$. We define a black-box knowledge extractor $\cal E$ for $\mathsf{FS[\Sigma]}$ as follows. ${\cal E}^{\cal A}$ simply works by running \smash{${\cal K}^{{\cal S}^{\cal A}}$}, where ${\cal S}^{\cal A}$ the adaptive \sigp adversary obtained by invoking Theorem~\ref{thmFS}. 
		For any subset $X \subseteq \cal X$, invoking the proof-of-knowledge property of $\mathsf{\Sigma}$ and using Theorem~\ref{thmFS}, we see that
		% \footnote{\cite{Dagdelen} claims that no successful black-box extractor exists against zero-knowledge \sigps (identification schemes with active security). However, their meta-reduction technique relies on the assumption that the extractor may not continue the simulation of the adversary after measuring one of its queries, since its internal state might have been disturbed too much by the measurement. Our main technical result shows this assumption to be unjustified.} 
		\allowdisplaybreaks
		\begin{align*}
		\Pr&\bigl[x\in X\wedge (x,w)\in R : (x,w)\leftarrow \mathcal{E}^{\mathcal{A}}\bigr] \\[1ex]
		&= \Pr\bigl[x\in X\wedge (x,w)\in R : (x,w)\leftarrow {\cal K}^{{\cal S}^{\cal A}} \bigr] \\[1.2ex]
		&= \frac{1}{p(\eta)}\cdot\Pr\bigl[x \in X \wedge v=accept  (x,v) \leftarrow \langle{\cal S}^{\cal A}, {\cal V}\rangle \bigr]^d- \kappa(\eta) \\
		&= \frac{1}{p(\eta)}\cdot\bigg(\sum_{x_\circ \in X} \Pr\bigl[x = x_\circ \wedge v=accept  (x,v) \leftarrow \langle{\cal S}^{\cal A}, {\cal V}\rangle \bigr]\bigg)^d- \kappa(\eta) \\
		&\geq \frac{1}{p(\eta)}\bigg(\frac{1}{O(q^2)} \sum_{x_\circ \in X}\Pr_H\bigl[ x = x_\circ \wedge V_{FS}^H(x,\pi): (x,\pi)\leftarrow \mathcal{A}^H \bigr] - \frac{1}{2q|{\cal C}|}\bigg)^d -\kappa(\eta)\\
		&\geq \frac{1}{p(\eta)\cdot O(q^{2d})}\cdot\Pr_H\left[ x\in X \wedge V_{FS}^H(x,\pi): (x,\pi)\leftarrow \mathcal{A}^H\right]^d - \mu(\eta)
		\end{align*}
		for some negligible function $\mu(\eta)$. \qed
	\end{proof}

	\begin{rem}\label{remUnruh} 
		We point out that in \cite{Unruh2017} Unruh considers a stronger notion of extractability than our Definition \ref{def:PoK}, where it is required that, in some sense, the extractor also recovers any local (possibly quantum) output of the adversary $\cal A$. In the light of Remark~\ref{rem:Generalization}, we expect that our result also applies to this stronger notion of extractability. 
		%In \cite{Unruh2017}, a definition of extractability is given that is stronger than the one we adopted in Definition \ref{defsecprop}. On a high level, the difference consist in that Unruh demands any quantum predicate on the internal state of the adversary to be preserved across the extraction, while we only require any property of $x$ (i.e\ some subset $X\subseteq \{0,1\}^{\ell_x}$) to be preserved. We note that Corollary \ref{corPoK} can be adapted to work for this stronger notion as well. If we substitute $(x,com,resp,\ket{\psi})$ for $(y,z)$ in Theorem \ref{thmmain} ($\ket{\psi}$ representing the internal state of the adversary), any quantum predicate $V(x,com,resp,\ket{\psi})$ is preserved by our \textit{malicious sigma-prover }{$\mathcal S$}. If now the definition of `proof of knowledge for a \sigp' is also adapted to require the knowledge extractor $\mathcal{K}$ to preserve this predicate, then Corollary \ref{corPoK} goes through as before.
	\end{rem}

	\section{Application to Fiat-Shamir signatures}\label{sec:sig}
	
	% \subsection{Unforgeability of Fiat-Shamir signatures}
	Any Fiat-Shamir non-interactive proof system can easily be transformed into a public-key signature scheme.\footnote{In fact, that is how the Fiat-Shamir transform was originally conceived in \cite{Fiat1987}. Only later \cite{Bellare1993} adapted the idea to construct a non-interactive zero-knowledge proof system.} The signer simply proves knowledge of a witness (the secret key) for a composite statement $x^* := x \| m$, which includes the public key $x$ as well as the message $m$. The signature $\sigma$ then consists of a proof for $x^*$. % More formally, $\mathsf{\Sigma^*} = ({\cal P^*,V^*})$ is a \sigp obtained from $\mathsf{\Sigma}$ by setting $\mathcal{P}^*(x\|m) = \mathcal{P}(x)$ and $\mathcal{V}^*(x\|m) = \mathcal{V}(x)$ for any $m$. 
	
	\begin{defi} \label{def:hardrelation}
		A binary relation $R$ with instance generator $G$ is said to be \emph{hard} if for any quantum polynomial-time algorithm $\mathcal{A}$ we have
		$$
		\Pr\left[ (x,w')\in R : (x,w)\leftarrow G, w' \leftarrow \mathcal{A}(x) \right]\leq \mu(\eta)
		$$
		for some negligible function $\mu(\eta)$, where $G$ is such that it always outputs a pair $(x,w)\in R$.
	\end{defi}
	
	\begin{defi}\label{defSig}
		A \emph{Fiat-Shamir signature scheme} based on a \sigp $\mathsf{\Sigma} = ({\cal P,V})$ for a hard relation $R$ with instance generator $G$, denoted by $\mathsf{Sig[\Sigma]}$ is defined by the triple {\it(Gen, Sign, Verify)}, with
		\begin{itemize}
			\item $\Gen$: Pick $(x,w) \leftarrow G$, set $sk:= (x,w)$ and $pk:= x$.
			\item $\Sign^H(sk,m)$: Return $(m,\sigma)$ where $\sigma\leftarrow P_{FS}^H(x\|m,w)$.
			\item $\Ver^H(pk,m,\sigma)$: Return $V_{FS}^H(x\|m,\sigma)$.
		\end{itemize}
		Here $(P_{FS}^H,V_{FS}^H) = \mathsf{FS[\Sigma^*]}$, where $\mathsf{\Sigma^*} = ({\cal P^*,V^*})$ is the \sigp obtained from $\mathsf{\Sigma}$ by setting $\mathcal{P}^*(x\|m) = \mathcal{P}(x)$ and $\mathcal{V}^*(x\|m) = \mathcal{V}(x)$ for any $m$. 
	\end{defi}    
	Note that by definition of $\mathsf{FS}$ in Section~\ref{sec:FStrans}, we use $V_{FS}^H(x\|m,\sigma)$ as shortcut for $V(x\|m,a,H(x\|m,a),z)$.
	
	We investigate the following standard security notions for signature schemes.
	\begin{defi}[$\mathsf{sEUF}\!-\!\mathsf{CMA/EUF}\!-\!\mathsf{NMA}$]\label{defsEUF}
		A signature scheme fulfills \emph{strong existential unforgeability under chosen-message attack {\upshape(}$\sEUFCMA${\upshape)}} if for all quantum polynomial-time algorithms $\mathcal{A}$ and for uniformly random $H:\mathcal{X'}\to \mathcal{C}$ it holds that
		$$\Pr\Bigl[\Ver^H(pk,m,\sigma) \wedge (m,\sigma)\notin \mathbf{Sig}\mathsf{-q} : (pk,sk)\leftarrow\Gen, (m,\sigma)\leftarrow \mathcal{A}^{H,\mathbf{Sig}}(pk)\Bigr]$$
		is negligible. Here $\mathbf{Sig}$ is classical oracle which upon classical input $m$ returns $\Sign^H(m,sk)$, and $\mathbf{Sig}\mathsf{-q}$ is the list of all queries made to $\mathbf{Sig}$.
		
		Analogously, a signature scheme fulfills \emph{existential unforgeability under no-message attack {\upshape(}$\EUFNMA${\upshape)}} if for all quantum polynomial-time algorithms $\mathcal{A}$ and for uniformly random $H:\mathcal{X'}\to \mathcal{C}$ it holds that
		$$\Pr\Bigl[\Ver^H(pk,m,\sigma) : (pk,sk)\leftarrow\Gen, (m,\sigma)\leftarrow \mathcal{A}^{H}(pk)\Bigr]$$
		is negligible.
	\end{defi}
	
	The unforgeability (against no-message attacks) of a Fiat-Shamir signature scheme is shown below to follow from the proof-of-knowledge property of the underlying proof system (hence, as we now know, of the underlying \sigp), under the assumption that the relation is hard, i.e.\ it is infeasible to compute $sk$ from $pk$.
	
	\begin{thm} \label{thm:NMA}
		Let $\mathsf{\Sigma}$ be \sigp for some hard relation $R$, with superpolynomially sized challenge space $\mathcal{C}$ and the proof-of-knowledge property according to Definition~\ref{def:PoK}. Then, the Fiat-Shamir signature scheme $\mathsf{Sig[\Sigma]}$ fulfills $\mathsf{EUF\!-\!NMA}$ security.
	\end{thm}
	\begin{proof}
		Let $\mathcal{A}$ be an adversary against $\mathsf{EUF\!-\!NMA}$, issuing at most $q$ quantum queries to $H$. We show that 
		\[
		\mathrm{Adv}^{\EUFNMA}_{\mathsf{Sig[\Sigma]}}(\mathcal{A}) := \Pr\left[\Ver^H(pk,m,\sigma) : (pk,sk)\leftarrow\Gen, (m,\sigma)\leftarrow \mathcal{A}^{H}(pk)\right]
		\]
		is negligible.
		
		Recall from Definition~\ref{defSig} of Fiat-Shamir signatures that the \sigp $\mathsf{\Sigma^*}$ is the \sigp $\mathsf{\Sigma}$ where the prover and verifier ignore the message part $m$ of the instance $x \| m$. A successful forgery $(m,\sigma)$ is such that $V_{FS}^H(x\|m,\sigma)$ accepts the proof $\sigma$. Therefore, 
		\begin{align} \label{eq:adv}
		\mathrm{Adv}^{\EUFNMA}_{\mathsf{Sig[\Sigma]}}(\mathcal{A}) = \E_{(x,w)\leftarrow G}\left[\Pr_H\left[V_{FS}^H(x\|m,\sigma): (m,\sigma) \leftarrow \mathcal{A}^H(x)\right]\right] \, .
		\end{align}
		
		Note that if $\mathsf{\Sigma}$ is a proof of knowledge, so is $\mathsf{\Sigma^*}$. Our Corollary~\ref{corPoK} assures that if $\mathsf{\Sigma^*}$ is a proof of knowledge, then also $\mathsf{FS[\Sigma^*]}$ is a proof of knowledge. 
		
		For fixed instance $x$, let $X$ be the set of instance/message strings $x'\|m$ where $x' = x$. We apply the knowledge extractor from Definition~\ref{def:PoK} to the adaptive FS-attacker $\mathcal{A}^H(x)$ that has $x$ hard-wired and outputs it along with a message $m$ and the proof/signature $\sigma$: There exists a knowledge extractor $\mathcal{E}$, constants $d,e$ and a polynomial $p$ (all independent of $x$) such that
		\begin{align} 
		\begin{split} \label{eq:PoK}
		\Pr_H&\left[ x'\|m\in X \wedge V_{FS}^H(x'\|m,\sigma): (x'\|m,\sigma)\leftarrow \mathcal{A}^H(x) \right] \\
		&\leq
		\left( \Pr\bigl[x'\|m \in X \,\wedge\, (x',w)\in R: (x'\|m,w)\leftarrow \mathcal{E}^{\mathcal{A}}\bigr] q^e p(\eta) + \mu(\eta) \right)^{1/d}
		\end{split}
		\end{align}
		
		Finally, taking the expected value of~\eqref{eq:PoK} over the choice of the instance $x$ according to the hard-instance generator $G$, we obtain that the left hand side equals $\mathrm{Adv}^{\EUFNMA}_{\mathsf{Sig[\Sigma]}}(\mathcal{A})$. For the right-hand side, we can use the concavity of $(\cdot)^{1/d}$ (note that we can assume without loss of generality that $d>1$) and apply Jensen's inequality to obtain
		\begin{align*}
		\E_{x \leftarrow G}& \left[ \left(  \Pr\bigl[x'\|m \in X \,\wedge\, (x',w)\in R: (x'\|m,w)\leftarrow \mathcal{E}^{\mathcal{A}}\bigr]q^e p(\eta) + \mu(\eta) \right)^{1/d} \right] \\
		&\leq
		\left( \E_{x \leftarrow G}  \Pr\bigl[x'\|m \in X \,\wedge\, (x',w)\in R: (x'\|m,w)\leftarrow \mathcal{E}^{\mathcal{A}}\bigr]q^e p(\eta) + \mu(\eta) \right)^{1/d} \, .
		\end{align*}
		Note that the expected probability is the success probability of the extractor to produce a witness $w$ matching the instance $x$. As long as the relation $R$ is hard according to Definition~\ref{def:hardrelation}, this success probability is negligible, proving our claim.
		
		\qed\end{proof}

	If we wish for unforgeability \textit{under chosen-message attack}, zero-knowledge is required as well. \cite{Unruh2017} and \cite{Kiltz2017} contain partial results that formalize this intuition, but they were unable to derive the extractability of the non-interactive proof system. Instead, they modify the \sigp to have a \emph{lossy mode}~\cite{AbdallaFLT12}, i.e. a special key-generation procedure that produces key pairs whose public keys are computationally indistinguishable from the real ones, but under which it is impossible for any (even unbounded) quantum adversary to answer correctly.
	
	Our new result above completes these previous analyses, so that we can now state precise conditions under which a \sigp gives rise to a (strongly) unforgeable Fiat-Shamir signature scheme, without the need for lossy keys.
	
	\begin{thm} \label{thm:CMA}
		Let $\mathsf{\Sigma}$ be \sigp for some hard relation $R$, with superpolynomially sized challenge space $\mathcal{C}$ and the proof-of-knowledge property according to Definition~\ref{def:PoK}. Assume further that $\mathsf{\Sigma}$ is $\varepsilon$-perfect (non-abort) honest-verifier zero-knowledge (naHVZK), has $\alpha$ bits of min entropy and computationally unique responses as defined in~\cite{Kiltz2017}. Then, $\mathsf{Sig[\Sigma]}$ fulfills $\sEUFCMA$ security.
	\end{thm}
	\begin{proof}
		By Theorem~3.3 of \cite{Kiltz2017}, we can use the naHVZK, min-entropy and computationally-unique-response properties of $\mathsf{\Sigma}$ to reduce an $\sEUFCMA$ adversary to an $\EUFNMA$ adversary\footnote{See also Theorem~25 in \cite{Unruh2017} for a different proof technique.}. The conclusion then follows immediately from our Theorem~\ref{thm:NMA} above.
		\qed \end{proof}

	\section{Extractable \sigps from quantum computationally unique responses}\label{seccollapsing}

	In the last section, we have seen that the proof-of-knowledge property of the underlying \sigp is crucial for a Fiat-Shamir signature scheme to be unforgeable. In \cite{Unruh2012}, Unruh proved that special soundness (a witness can be constructed efficiently from two different accepting transcripts) and perfect unique responses are sufficient conditions for a \sigp to achieve this property in the context of quantum adversaries. The perfect-unique-responses property is used to show that the final measurement of the \sigp adversary that produces the response is nondestructive conditioned on acceptance. This property ensures that the extractor can measure the response, and then rewind ``as if nothing had happened''.
	
	A natural question is therefore which other property except the arguably quite strict condition of perfect unique responses is sufficient to imply extractability together with special soundness. In \cite{Ambainis2014}, the authors show that computationally unique responses is insufficient to replace perfect unique responses.  A \sigp has computationally unique responses if the verification relation $V$ is collision-resistant from responses to commitment-challenge pairs in the sense that it is computationally hard to find two valid responses for the same commitment-challenge pair.
	
	In \cite{Unruh2016}, Unruh introduced the notion of collapsingness, a quantum generalization of the collision-resistance property for hash functions. It is straight-forward to generalize this notion to apply to binary relations instead of just functions.
	
	\begin{defi}[generalized from \cite{Unruh2016}] \label{def:collapsing}
		Let $R:\mathcal X\times\mathcal Y\to\{0,1\}$ be a relation with $|X|$ and $|Y|$ superpolynomial in the security parameter $\eta$, and define the following two games for polynomial-time  two-stage adversaries $\mathcal A=(\mathcal A_1,\mathcal A_2)$,
		\begin{center}
			\begin{tabular}{lcl}
				$\mathrm{Game\ 1:}$\switch{\,}{&&\\}
				$(S,X,Y)\leftarrow\mathcal A_1$, $r\leftarrow R(X,Y)$,&$X\leftarrow\mathcal M(X)$,& $Y\leftarrow\mathcal M(Y)$, $b\leftarrow\mathcal A_2(S,X,Y)$\\[2mm]
				$\mathrm{Game\ 2:}$\switch{\,}{&&\\}
				$(S,X,Y)\leftarrow\mathcal A_1$, $r\leftarrow R(X,Y)$,&&$Y\leftarrow\mathcal M(Y)$, $b\leftarrow\mathcal A_2(S,X,Y)$.
			\end{tabular} 
		\end{center}
		Here, $X$ and $Y$ are registers of dimension $|X|$ and $|Y|$, respectively, $\mathcal M$ denotes a measurement in the computational basis, and applying $R$ to quantum registers is done by computing the relation coherently and measuring it. $R$ is called \emph{collapsing from $\mathcal X$ to $\mathcal Y$}, if an adversary cannot distinguish the two experiments if the relation holds, i.e. if for all adversaries $\mathcal A$ it holds that 
		\begin{equation}
		\left|\Pr_{\mathcal A,\ \mathrm{Game\ 1}}\left[r=b=1\right]-\Pr_{\mathcal A,\ \mathrm{Game\ 2}}\left[r=b=1\right]\right|\le\mathrm{negl}(\eta).
		\end{equation}
	\end{defi}
	Note that this definition is equivalent to Definition 23 in \cite{Unruh2016} for functions, i.e. if $R(x,y)=1$ if and only if $f(x)=y$ for some function $f$.
	
	Via the relation that is computed by the second stage of the verifier, the collapsingness property can be naturally  defined for \sigps.
	\begin{defi}[Quantum computationally unique responses]\label{def:qcur}
		%with commitment, challenge and response spaces $\mathcal Y$, $\mathcal C$ and $\mathcal Z$ and honest verifier $\mathcal V$ 
		A \sigp has \emph{quantum computationally unique responses}, if the verification predicate  $V(x,\cdot,\cdot,\cdot): \mathcal Y\times \mathcal C\times \mathcal Z\to\{0,1\}$ seen as a relation between $ \mathcal Y\times \mathcal C$ and $ \mathcal Z$ is collapsing from $\mathcal Z$ to $\mathcal Y\times \mathcal C$, where $\mathcal Y$, $\mathcal C$ and $\mathcal Z$ are the commitment, challenge and response spaces of the protocol, respectively.
	\end{defi}
	Intuitively, for fixed commitment-challenge pairs, no adversary should be able to determine whether a superposition over successful responses $z$ has been measured or not. As in the case of hash functions (where collapsingness is a natural stronger quantum requirement than collision-resistance), quantum computationally unique responses is a natural stronger quantum requirement than computationally unique responses. 
	
	The following is a generalization of Theorem 9 in \cite{Unruh2012} where the assumption of perfect unique responses is replaced by the above quantum computational version. Additionally, we relax the special soundness requirement to {\em $t$-soundness}, which requires that for any first message $a$, for uniformly random chosen challenges $c_1,\ldots,c_t$, and for any responses $z_1,\ldots,z_t$ with $V(x,a_i,c_i,z_i)$ for all $i \in \{1,\ldots,t\}$, a witness $w$ for $x$ can be efficiently computed except with negligible probability (over the choices of the $c_i$). 
	
	% As discussed above, Theorem 9 in \cite{Unruh2012} uses the prerequisite of perfect unique responses exclusively to ensure that the second run of the prover, after rewinding, succeeds. This success of the prover after rewinding is, however, clearly also implied by the collapsingness property, i.e. we have the following
	
	\begin{thm}[Generalization of Theorem 9 from \cite{Unruh2012}]\label{thm:quantum-cur2PoK}
		Let $\Pi$ be a \sigp with $t$-soundness for some constant $t$ and with quantum computationally unique responses. Then $\Pi$ is a computational proof of knowledge as in Definition~\ref{def:PoK}. 
	\end{thm}
	The proof follows very much the proof of Theorem 9 in \cite{Unruh2012}, up to some small extensions; thus, we only give a proof sketch here. 
	
	\begin{proof}[sketch]
		We consider the following extractor $\cal K$. It runs $\cal A$ to the point where it outputs $a$. Then, it chooses a random challenge $c_1$ and sends it to $\cal A$, and obtains a response $z_1$ by measuring $\cal A$'s corresponding register. $\cal K$ then rewinds $\cal A$ (on the measured state!) and chooses and sends to $\cal A$ a fresh random challenge $c_2$, resulting in a response $z_2$, etc., up to obtaining response $z_t$. If $V(x,a_i,c_i,z_i)$ for all $i \in \{1,\ldots,t\}$ then $\cal K$ can compute $w$ except with negligible probability by the $t$-soundness property; otherwise, it aborts.  
		
		It remains to analyze the probability, denoted by $F$ below, that $V(x,a_i,c_i,z_i)$ for all $i$. If the \sigp has {\em perfect} unique responses then measuring the response $z$ is {\em equivalent} to measuring whether the response satisfies the verification predicate $V$ (with respect to $x,a,c$). Lemma~\ref{lemma:FvsV} in Appendix~\ref{app:ProjBound}, which generalizes Lemma 7 in \cite{Unruh2012}, allows us then to control the probability $F$ by means of the probability $V$ that $\cal A$ succeeds in convincing the verifier in an ordinary run (this holds for an arbitrary but fixed $a$, and on average over $a$ by means of Jensen's inequality). 
		If the \sigp has {\em quantum computationally} unique responses instead, then measuring the response $z$ is {\em computationally indistinguishable} from measuring whether the response satisfies the verification predicate, and so there can only be a negligible loss in the success probability of $\cal K$ compared to above. \qed 
	\end{proof}

	We expect the above theorem to be very useful in practice, for the following reason. Usually, \sigps deployed in Fiat-Shamir signature schemes have computationally unique responses to ensure strong unforgeability via Theorem \ref{thm:CMA} or similar reductions. On the other hand, only very artificial separations between the notions of collision resistance and collapsingness for hash functions are known (e.g. the one presented in~\cite{Zhandry2017}). It is therefore plausible that many \sigps deployed in strongly unforgeable Fiat-Shamir signature schemes have quantum computationally unique responses as well. In the next section we take a look at a couple of examples that form the basis of some signature schemes submitted to the NIST competition for the standardization of post-quantum cryptographic schemes.

	\section{Application to NIST submissions}\label{secNIST}
	In the previous sections we gave sufficient conditions for a Fiat-Shamir signature scheme to be existentially unforgeable in the QROM. Several schemes of the Fiat-Shamir kind have made it into the second round of the NIST post-quantum standardization process. In this section we outline how our result might be applied to some of these schemes, and under which additional assumptions. We leave the problems of applying our techniques to the actual (highly optimized) signature schemes and of working out the concrete security bounds  for future work.
	
	\subsection{\textsf{Picnic}}
	In order to obtain QROM-security, \textsf{Picnic} uses the \emph{Unruh transform} \cite{Unruh2015} instead of the Fiat-Shamir transformation, incurring a 1.6x loss in efficiency (according to \cite{Chase2017}) compared to \textsf{Fish}, which is the same scheme under plain Fiat-Shamir.

	The underlying sigma-protocol for these schemes is \textsf{ZKB++} \cite{Chase2017}, an optimized version of \textsf{ZKBoo} \cite{GMO}, which uses an arbitrary one-way function $\phi$, a commitment scheme $\mathsf{COM}$ and a multi-party computation protocol to prove knowledge of a secret key. Roughly, a prover runs the multi-party protocol `in its head' (i.e.\ simulates the three agents from the protocol, see \cite{Ishai2007}) to compute $pk:= \phi(sk)$. Only a prover who knows the secret key can produce the correct view of all three agents, but the public key suffices to verify the correctness of two of the views. In the first round, the prover uses $\mathsf{COM}$ to commit to all three views separately, and sends these commitments to the verifier. The verifier replies with a random challenge $i\in \{1,2,3\}$, to which the prover in turn responds by opening the $i$-th and $i\!+\!1$-th commitment.
	
	$\mathsf{ZKBoo}$ does not specify a concrete commitment scheme for $\mathsf{COM}$. A natural option is to commit by hashing the input together with some random bits.
	\begin{cor}\label{cor:ZKBoo}
		$\mathsf{Sig[ZKBoo]}$ is strongly existentially unforgeable in the QROM when $\mathsf{COM}$ is instantiated with a collapsing hash function $H$.\footnote{We could in principle decide to model $H$ as a random oracle as well, (since we are already in the QROM anyway), in which case the collapsingness of $H$ would follow from~\cite{Unruh2016}. However, we formalized security of \sigps (soundness \& PoK) {\em in the standard model}. We could extend these definitions to work in the QROM, but then Lemma \ref{lemma:StatVsAaptSound} and \ref{lemma:StatVsAaptPoK} would no longer follow. Therefore, we choose to simply assume collapsingness of $H$.}
	\end{cor}
	\begin{proof}
		 Since the response of the prover in the third round consists only of openings to the commitments $c_i,c_{i\!+\!1}$, i.e. preimages of $c_i$ and $c_{i\!+\!1}$ under $H$, and since collapsingness is closed under concurrent composition \cite{Fehr2018}, the collapsingness of $H$ implies that \textsf{ZKBoo} has quantum computational unique responses. $\mathsf{ZKBoo}$ further has $3$-soundness, and thus the claim follows using Theorems \ref{thm:quantum-cur2PoK} and  \ref{thm:CMA}.
		\qed\end{proof}
	
	\textsf{ZKB++} improves on \textsf{ZKBoo} by introducing optimizations specific to the signature context, which complicate the analysis of the overall scheme. We therefore leave the adaption of Corollary \ref{cor:ZKBoo} to \textsf{ZKB++} and \textsf{Fish} for future work.
	
	We also point out that \textsf{Picnic2} (a later version of \textsf{Picnic}) is not $t$-sound because a witness can be computed from 3 responses only under certain restrictions on the challenges. However, this can be taken care of by a variation of the $t$-soundness property, as proven in Lemma~\ref{lemma:FvsV2} in Appendix~\ref{app:ProjBound}. 
	
	%We would like to make two remarks about the assumptions in the above corollary. Firstly, it is standard to treat SHA-256, which is used in the actual specification of \textsf{ZKBoo}, like a collision resistant hash function. Considering the discussion in the end of Section \ref{seccollapsing}, it is therefore natural to treat it like a collapsing hash function as well. And secondly, the Fiat-Shamir transformation uses the QROM anyway, so it does not constitute an additional assumption to treat the hash used for commitments as a quantum-accessible random oracle. 
	\subsection{Lattice-based Fiat-Shamir signature schemes  --  \textsf{CRYSTALS-Dilithium} and \textsf{qTesla}}
	In \cite{Lyubashevsky2009} and \cite{Lyubashevsky2012}, Lyubashevsky developed a Fiat-Shamir signature scheme based on (ring) lattice assumptions. In the following, we explain the lattice case and mention ring-based lattice terms in parentheses. The underlying sigma protocol, which forms the basis of the NIST submissions \textsf{CRYSTALS-Dilithium} and \textsf{qTesla}, can be roughly described as follows. The instance is given by a key pair $((A, T), S)$, with $T=AS$. Here, $A$ and $S$ are matrices of appropriate dimensions over a finite field (polynomials of appropriate degree), and $S$ is \emph{small}. For the first message to the verifier, the prover selects a random short vector (small polynomial) $y$, and sends over $Ay$. The second message, from the verifier to the prover, is a random vector (polynomial) $c$ with entries (coefficients) in $\{-1,0,1\}$ and a small Hamming weight. The third message, i.e. the response of the prover, is $z=Sc+y$, which is short (small) as well. The prover actually sends 
	$z$ only with a particular probability, which is chosen so as to make the distribution of (sent) $z$ independent of $S$. Otherwise, it aborts and tries again. Verification is done by checking whether $z$ is indeed short (small), and whether $Az-Ay = Tc$. Let us denote this protocol 
	by $\mathsf{Lattice\Sigma}$. In the following we restrict our attention to the lattice case, but we expect that one can do a similar analysis for the ring-based schemes.
	
	The security of the scheme is, in the lattice case, based on the \textsc{Short Integer Solution} (\textsc{SIS}) problem, which essentially guarantees that it is hard to find an integral solution to a linear system that has a small norm. The computationally unique responses property for the simple \sigp described above, in fact, follows directly from \textsc{SIS}: If one can find a vector $c$ and two short vectors $x_i$, $i=1,2$ such that $Ax_0=c=Ax_1$, then the difference $x=x_1-x_0$ is a short solution to the linear system $Ax=0$.
	
	Another way to formulate the computationally unique responses property for the above \sigp is as follows. Let $S\subset \mathbb F_q^n$ be the set of short vectors. 
	Let $f_A: S\to \mathbb F_q^m$ be the restriction to $S$ of the linear map given by the matrix $A\in\mathbb F_q^{m\times n}$. The \sigp above has computationally unique responses if and only if $f_A$ is collision resistant. 
	As pointed out at the end of Section \ref{seccollapsing}, the known examples that separate the collision resistance and collapsingness properties are fairly artificial. Hence it is a natural to assume that $f_A$ is collapsing as well. 
	
	\begin{ass}\label{ass-colSIS}
		For $m,n$ and $q$ polynomial in the security parameter $\eta$, the function family $f_A$ keyed by a uniformly random matrix $A\in\mathbb F_q^{m\times n}$ is collapsing.
	\end{ass}
	Under Assumption \ref{ass-colSIS}, $\mathsf{Lattice\Sigma}$ has quantum computational unique responses, and hence gives rise to an unforgeable Fiat-Shamir signature scheme.
	\begin{cor}
		Under Assumption \ref{ass-colSIS}, $\mathsf{Sig[Lattice\Sigma]}$ is strongly existentially unforgeable in the QROM.
	\end{cor}
	As mentioned at the end of the introduction, 
	% In a concurrent and independent work, which became publicly available after the first version of this work, 
	in their concurrent and independent work \cite{LZ19}, Lie and Zhandry show that $f_A$ satisfies their notion of \emph{weak}-collapsingness (assuming hardness of LWE), which roughly says that there is some non-negligible probability that the adversary \emph{does not} notice a measurement. Weak-collapsingness implies a similarly weakened variant of our property `quantum computational responses', which is still sufficient to let the proof of \mbox{Theorem \ref{thm:quantum-cur2PoK}} go through, albeit with a worse but still non-negligible success probability for the knowledge-extractor.

	\section{Acknowledgement}
	We thank Tommaso Gagliardoni and Dominique Unruh for comments on early basic ideas of our approach, and Andreas H\"ulsing, Eike Kiltz and Greg Zaverucha for helpful discussions. We thank Thomas Vidick for helpful remarks on an earlier version of this article.

	JD and SF were partly supported by the EU Horizon 2020 Research and Innovation
	Program Grant 780701 (PROMETHEUS).  JD, CM, and CS were supported by a NWO VIDI grant (Project No. 639.022.519). During finalization of this work JD has been partially funded by
	ERC-ADG project 740972 (ALGSTRONGCRYPTO).

	\bibliographystyle{alpha}
	\bibliography{QROM}

	\begin{appendix}
		
		\section{Proof of Lemma~\ref{lemma:StatVsAaptSound} and \ref{lemma:StatVsAaptPoK}} \label{app:proofs}
		
		\begin{proof}[of Lemma~\ref{lemma:StatVsAaptSound}] 
			Let ${\cal A}$ be an adaptive \sigp adversary, producing $x$ and $a$ in the first stage, and $z$ in the second stage. We then consider the following algorithms. ${\cal A}_{init}$ runs the first stage of ${\cal A}$ (using the same initial state), outputting $x$ and~$a$. Let $\ket{\psi_{x,a}}$ be the corresponding internal state at this point. Furthermore, for any possible $x$ and~$a$, ${\cal A}_{x,a}$ is the following static \sigp adversary. Its initial state is $\ket{\psi_{x,a}}\ket{a}$ and in the first stage it simply outputs $a$, and in the second stage, after having received the verifier's challenge, it runs the second stage of~${\cal A}$. 
			We then see that 
			\begin{align*}
			\Pr\bigl[x \not\in {\cal L} &\,\wedge\, v = accept :(x,v) \leftarrow \langle{\cal A}, {\cal V}\rangle \bigr] \\[1.2ex]
			&= \sum_{x_\circ \not\in {\cal L}} \Pr\bigl[x = x_\circ \,\wedge\, v = accept :(x,v) \leftarrow \langle{\cal A}, {\cal V}\rangle \bigr]\\
			&= \sum_{x_\circ \not\in {\cal L}} \sum_a \Pr\bigl[{\cal A}_{init} = (x_\circ,a)\bigr] \Pr\bigl[\langle{\cal A}_{x_\circ,a} , {\cal V}(x_\circ) \rangle  = accept \bigr] \, .
			\end{align*}
			Since $ \Pr\bigl[\langle{\cal A}_{x_\circ,a} , {\cal V}(x_\circ) \rangle = accept \bigr]$ is bounded by a negligible function, given that ${\cal A}_{x,a}$ is a (quantum polynomial-time/unbounded) static adversary, the claim follows. \qed
		\end{proof}
		
		\begin{proof}[of Lemma~\ref{lemma:StatVsAaptPoK}]
			Let ${\cal A}$ be an adaptive \sigp adversary, producing $x$ and $a$ in the first stage, and $z$ in the second stage. We construct a black-box knowledge extractor ${\cal K}_{ad}$ that works for any such ${\cal A}$. In a first step, ${\cal K}_{ad}^{\cal A}$ runs the first stage of ${\cal A}$ using the black-box access to ${\cal A}$ (and having access to the initial state of ${\cal A}$). Below, we call this first stage of ${\cal A}$ as ${\cal A}_{init}$. 
			This produces $x$ and $a$, and we write $\ket{\psi_{x,a}}$ for the corresponding internal state. Then, it runs ${\cal K}_{na}^{{\cal A}^{x,a}}$, where ${\cal K}_{na}$ is the knowledge extractor guaranteed to exist for static adversaries, and ${\cal A}^{x,a}$ is the static adversary that works as follows. It's initial state is $\ket{\psi_{x,a}}\ket{a}$ and in the first stage it simply outputs $a$, and in the second stage it runs the second stage of ${\cal A}$ on the state $\ket{\psi_{x,a}}$. Note that having obtained $x$ and $a$ and the state $\ket{\psi_{x,a}}$ as first step of ${\cal K}_{ad}^{\cal A}$, ${\cal K}_{na}^{{\cal A}^{x,a}}$ can then 
			be executed with black box access to (the second stage of)~${\cal A}$. 
			For any subset $X \subseteq {\cal X}$, we now see that
			\allowdisplaybreaks
			\begin{align*}
			\Pr&\left[x\in X\wedge (x,w)\in R : (x,w)\leftarrow \mathcal{K}_{ad}^{\mathcal{A}}\right] \\[0.5em]
			&= \sum_{x\in X} \sum_a \Pr\bigr[{\cal A}_{init} = (x,a) \bigl] \Pr\Bigl[(x,w)\in R : w\leftarrow \mathcal{K}_{na}^{\mathcal{A}^{x,a}} \Bigr] \\
			&\geq \sum_{x\in X} \sum_a  \Pr\bigr[{\cal A}_{init} = (x,a) \bigl] \cdot\frac{1}{p(\eta)} \cdot \Pr\bigl[\langle{\cal A}^{x,a} , {\cal V}(x)\rangle = accept \bigr]^d- \kappa(\eta)\\
			&\geq \frac{1}{p(\eta)}\bigg(\sum_{x\in X} \sum_a \Pr\bigr[{\cal A}_{init} = (x,a) \bigl] \Pr\bigl[\langle{\cal A}^{x,a} , {\cal V}(x)\rangle = accept \bigr]\bigg)^d -\kappa(\eta)\\
			&= \frac{1}{p(\eta)} \Pr\bigl[ x \in X \wedge v\!=\!1: (x,v)\leftarrow \langle{\cal A}^{x,a} , {\cal V}(x)\rangle \bigr]^d -\kappa(\eta) \, ,
			\end{align*}
			where the first inequality is because of the static proof-of-knowledge property, and the second is Jensen's inequality, noting that we may assume without loss of generality that $d \geq 1$. \qed
		\end{proof}

		\section{Generalization of Lemma 7 from \cite{Unruh2012}} \label{app:ProjBound}
		
		The following is a generalization of Lemma 7 from \cite{Unruh2012}. It relates the success probability of applying a random projection to a state vector with the success probability of sequentially applying $t$ random projections, where ``success probability'' here is in terms of the (average) square-norm of the projected state vector. This statement gives us the means to relate the probability of a interactive prover making the verifier accept with the probability of an extractor making the verifier accept $t$ times, when rewinding $t-1$ times and using a freshly random (and independent) challenge each time. 
		
		\begin{lem}\label{lemma:FvsV}
			Let $P_1,\ldots,P_n$ be projections and $\ket{\psi}$ a state vector, and set 
			$$
			V := \frac{1}{n} \sum_i \bra{\psi} P_i \ket{\psi} = \frac{1}{n} \sum_i \|P_i \ket{\psi}\|^2
			\qquad\text{and}\qquad
			F := \frac{1}{n^t}\sum_{i_1 \cdots i_t} \|P_{i_t} \cdots P_{i_1} \ket{\psi} \|^2 \, .
			$$
			Then $F \geq V^{2t-1}$. 
		\end{lem}
		The case $t=2$ was proven in \cite[Lemma 7]{Unruh2012}. We show here how to extend the proof to $t = 3$; the general case works along the same lines. 
		
		\begin{proof}[of the case $t = 3$]
			For convenience, set $A:= \frac{1}{n} \sum_i P_i$ and $\ket{\psi_{ijk}} := P_k P_j P_i \ket{\psi}$. Then, using convexity of the function $x \mapsto x^5$ to argue the first inequality, we get
			\switch{
				\begin{align*}
				V^5 = (\bra{\psi} A \ket{\psi})^5 \leq \bra{\psi} A^5 \ket{\psi} = \frac{1}{n^5} \sum_{i j k \ell m} \bra{\psi} P_i P_j P_k P_\ell P_m \ket{\psi} = \frac{1}{n^5} \sum_{i j k \ell m} \braket{\psi_{ijk}}{\psi_{m \ell k}} \\
				= \frac{1}{n} \sum_k \bigg( \frac{1}{n^2} \sum_{ij} \bra{\psi_{ijk}}\bigg)\bigg(\frac{1}{n^2} \sum_{\ell m} \ket{\psi_{m \ell k}}\bigg) = \frac{1}{n} \sum_k \Big\| \frac{1}{n^2} \sum_{ij} \ket{\psi_{ijk}}\Big\|^2 \leq \frac{1}{n^3} \sum_{ijk} \| \ket{\psi_{ijk}}\|^2 = F \, ,
				\end{align*}
			}
			{
				\begin{align*}
				V^5 = &(\bra{\psi} A \ket{\psi})^5 = \bra{\psi} A^5 \ket{\psi} = \frac{1}{n^5} \sum_{i j k \ell m} \bra{\psi} P_i P_j P_k P_\ell P_m \ket{\psi}\\
				& = \frac{1}{n^5} \sum_{i j k \ell m} \braket{\psi_{ijk}}{\psi_{m \ell k}} 
				= \frac{1}{n} \sum_k \bigg( \frac{1}{n^2} \sum_{ij} \bra{\psi_{ijk}}\bigg)\bigg(\frac{1}{n^2} \sum_{\ell m} \ket{\psi_{m \ell k}}\bigg)\\
				&\hspace{20pt} = \frac{1}{n} \sum_k \Big\| \frac{1}{n^2} \sum_{ij} \ket{\psi_{ijk}}\Big\|^2 \leq \frac{1}{n^3} \sum_{ijk} \| \ket{\psi_{ijk}}\|^2 = F \, ,
				\end{align*}
			}
			where the last inequality is Claim~2 in the proof of Lemma 7 in \cite{Unruh2012}. \qed 
		\end{proof}
		The following is a generalization of Lemma 7 from \cite{Unruh2012} in a different direction. It gives us control over the success probability of the extractor when the challenge consists of two parts, and the extractor works by rewinding once with a freshly chosen challenge pair, and once more where now one part of the challenge is re-used and only the other part is freshly chosen. 
		
		\begin{lem}\label{lemma:FvsV2}
			Let $P_{ij}$ ($1\leq i \leq n$, $1\leq j \leq m$)  be projections $\ket{\psi}$ a state vector, and set 
			$$
			V := \frac{1}{n m} \sum_{i,j} \|P_{i,j} \ket{\psi}\|^2
			\qquad\text{and}\qquad
			F := \frac{1}{n^2 m^3}\sum_{i_1,i_2 \atop j_1,j_2,j_3} \|P_{i_2 j_3} P_{i_2 j_2} P_{i_1 j_1} \ket{\psi} \|^2 \, .
			$$
			Then $F \geq V^6$. 
		\end{lem}
		
		\begin{proof}
			We set $\ket{\varphi_{i_1 j_1}} := P_{i_1 j_1}\ket{\psi}/\bra{\psi}P_{i_1 j_1}\ket{\psi}$. Then
			\begin{align*}
			F &= \frac{1}{n^2 m^3}\sum_{i_1,i_2 \atop j_1,j_2,j_3} \|P_{i_2 j_3} P_{i_2 j_2} P_{i_1 j_1} \ket{\psi} \|^2 \\
			&= \frac{1}{n^2 m}\sum_{i_1,i_2, j_1} \frac{1}{m^2}\sum_{j_2,j_3} \|P_{i_2 j_3} P_{i_2 j_2} \ket{\varphi_{i_1 j_1}} \|^2 \, \bra{\psi}P_{i_1 j_1}\ket{\psi} \\
			&\geq \frac{1}{n^2 m}\sum_{i_1,i_2, j_1} \bigg(\frac{1}{m}\sum_{j_2} \| P_{i_2 j_2} \ket{\varphi_{i_1 j_1}} \|^2\bigg)^3 \, \bra{\psi}P_{i_1 j_1}\ket{\psi} && \text{(Lemma~\ref{lemma:FvsV})} \\
			&= \frac{1}{n^2 m}\sum_{i_1,i_2, j_1} \bigg(\frac{1}{m}\sum_{j_2} \| P_{i_2 j_2}P_{i_1 j_1} \ket{\psi} \|^2/\bra{\psi}P_{i_1 j_1}\ket{\psi}^{2/3}\bigg)^3 \\
			&\geq \bigg(\frac{1}{n^2 m^2}\sum_{i_1,i_2, j_1,j_2} \| P_{i_2 j_2}P_{i_1 j_1} \ket{\psi} \|^2/\bra{\psi}P_{i_1 j_1}\ket{\psi}^{2/3}\bigg)^3 && \text{(Jensen's inequality)} \\
			&\geq \bigg(\frac{1}{n^2 m^2}\sum_{i_1,i_2, j_1,j_2} \| P_{i_2 j_2}P_{i_1 j_1} \ket{\psi} \|^2\bigg)^3 \\
			&\geq \bigg(\frac{1}{n m}\sum_{i_1,j_1} \| P_{i_1 j_1} \ket{\psi} \|^2\bigg)^6 && \text{(Lemma~\ref{lemma:FvsV})} 
			\end{align*}
			This proves the claim. \qed
		\end{proof}

	\end{appendix}
	
\end{document}